\documentclass[a4paper]{article}

\usepackage{a4wide}
\usepackage[UKenglish]{babel}
\usepackage{graphicx,subfigure}
	\graphicspath{{./}{figures/}}
\usepackage[colorlinks=true,linkcolor=blue,citecolor=red]{hyperref}
\usepackage{xcolor}
\usepackage{amsfonts,amsmath,amssymb,amsthm,mathrsfs,mathtools,bbold,empheq}
\usepackage[inline]{enumitem}
\usepackage[ruled,vlined,linesnumbered]{algorithm2e}
\usepackage{caption}
\usepackage[normalem]{ulem}
\usepackage{geometry}

\newcommand{\abs}[1]{\left\lvert#1\right\rvert}
\DeclarePairedDelimiter\ave{\langle}{\rangle}
\newcommand{\cI}{\mathcal{I}}
\renewcommand{\P}{\operatorname{Prob}}
\newcommand{\R}{\mathbb{R}}
\newcommand{\zb}{\boldsymbol{z}}
\newcommand{\Zb}{\boldsymbol{Z}}
\newcommand{\Yb}{\boldsymbol{Y}}

\newcommand{\yb}{\boldsymbol{y}}

\newcommand{\pr}[1]{{}^\prime\!#1}

\allowdisplaybreaks

\newtheorem{theorem}{Theorem}[section]

\theoremstyle{remark}\newtheorem{remark}[theorem]{Remark}
\newtheorem{Assumption}[theorem]{Assumption}

\begin{document}
%\title{General kinetic models with transition probabilities for systems of interacting agents}
\title{Kinetic models for systems of interacting agents with multiple microscopic states}
\author{Marzia Bisi \thanks{Department of Mathematical, Physical and Computer Sciences, University of Parma, Italy, (\texttt{marzia.bisi@unipr.it})} \and Nadia Loy \thanks{Department of Mathematical Sciences ``G. L. Lagrange'', Politecnico di Torino, Corso Duca degli Abruzzi 24, 10129 Torino, Italy (\texttt{nadia.loy@polito.it})}
}\date{\small }

\maketitle

%\textcolor{red}{Da fare: commento su gas}

%\textcolor{red}{Non mi convince ancora il titolo, pensiamoci}
%\textcolor{blue}{MB: toglierei ''simultaneously evolving''}

\begin{abstract}
We propose and investigate general kinetic models
%of Boltzmann type
with transition probabilities that can describe the simultaneous change of multiple microscopic states of the interacting agents. These models can be applied to many problems in socio-economic sciences, where individuals may change both their compartment and their characteristic kinetic variable, as for instance kinetic models for epidemics or for international trade with possible transfers of agents. Mathematical properties of our kinetic model are proved, as existence and uniqueness of a solution for the Cauchy problem in suitable Wasserstein spaces. The quasi-invariant asymptotic regime, leading to simpler kinetic Fokker-Planck-type equations, is investigated and commented on in comparison with other existing models. Some numerical tests are performed in order to show time evolution of distribution functions and of meaningful macroscopic fields, even in case of non-constant interaction probabilities.
\end{abstract}


\section{Introduction}

In the literature of kinetic models for multi-agent systems, there is an increasing interest in phenomena where the agents are characterized by a multiple microscopic state and they are, in particular, divided into subpopulations.

Kinetic theory for multi-agent systems has its roots in the classical kinetic theory related to the Boltzmann equation for the description of a rarefied gas \cite{cercignani1988BOOK}, in which individuals are molecules identified by a microscopic state $v$ that is the velocity and that changes because of \textit{binary interactions}. The classical kinetic theory for gas dynamics has been generalized to various kinds of interacting systems, where the microscopic state is not necessarily the velocity, providing
reasonable mathematical models for many socio-economic problems, as the evolution of wealth distribution~\cite{cordier2005JSP, toscani2019PRE}, the opinion formation~\cite{toscani2006CMS,toscani2018PRE}, the pedestrian or vehicular traffic dynamics \cite{festa2018KRM,freguglia2017CMS}, birth and death processes \cite{greenman2016PRE,loy2021KRM} and many others. Also in this field, models describing the interaction of different populations through a system of Boltzmann equations have been proposed for instance in \cite{DB2008CMS, During2010Econ} for wealth exchanges, in~\cite{during2009PRSA} for opinion formation in presence of leaders, in \cite{borsche2022PhysA} for multilane traffic models.

The classical Boltzmann equation has been extended some decades ago to mixtures of different gaseous species~\cite{chapman1970BOOK, kogan1969BOOK}, even in presence of chemical reactions \cite{Gio1999}, and also consistent BGK approximations have been proposed and investigated \cite{andries2002JSP, bisi2010PRE, pirner2021Fluids}. In these models, when describing bimolecular chemical reactions, a given binary interaction between molecules may simultaneously cause both the change of the velocity and transfers of the involved molecules to different species. Also models in which the molecule is characterized by the belonging to a species, a molecular velocity and another inner variable modelling the internal energy have been proposed \cite{borsoni2022CIMP}.
%On the other hand, several kinetic models for socio-economic sciences consider exchanges of goods within a single population.

In the field of kinetic equations for multi-agent systems applied to socio-economic phenomena, a very interesting class of models in which individuals have a multiple microscopic state corresponds to models in which the total population is divided into subgroups and each agent is also described by a physical quantity (wealth, opinion, viral load, etc.). Each group is then characterized by a distribution function depending on the given microscopic physical quantity that can be exchanged both with individuals of the same subgroup and with individuals of a different subgroup, according to suitable interaction rules.
For simplicity, interactions causing exchanges of goods and the ones giving rise to a change of subgroup of one agent are often modelled separately, by means of different kinetic operators \cite{During2010Econ, DB2008CMS, during2009PRSA, dimarco2020PRE, machadoramos2019JMB, dellamarca2022MMAS}.
 Formally, models belonging to this class have been derived in \cite{loy2021KRM}, where the authors describe multi-agent systems in which the agents are characterized by a double microscopic state: a physical quantity~$v$, changing because of binary interactions, and a label $x$ denoting the subgroup of the agent, changing as a consequence of a Markovian process. The two stochastic processes for the evolution of $v$ and $x$ are \textit{independent} and occur with different frequencies, giving thus rise to different operators, where the one relevant to the variation of $v$ may be the sum of \textit{inter-} or \textit{intra-} species interactions.

However, a class of models worth to be investigated is the one in which the agent changes the microscopic quantity $v$ and the subgroup simultaneously as a consequence of the \textit{same} binary interaction. An example is given by the aforementioned bimolecular chemical reactions in gas mixtures.
The classical kinetic description for chemically reacting gases was introduced in the pioneering work~\cite{rossani1999PhysA} where each gas species has a different distribution function of the microscopic velocity of its molecules. A pair of molecules belonging to potentially different species may interact exchanging both the velocity and the species. Such bimolecular reaction is described by means of a collisional operator that involves both distributions of reactants and products, and the change of the velocity is included by means of a transformation  with unit Jacobian like in the classical strong Boltzmann equation.
In the literature of kinetic theory for socio-economic sciences, the recent paper \cite{bisi2021PhTB} describes the trade among different subpopulations, living in different countries, taking into account also possible transfers of individuals from one country to another, by means of suitable Boltzmann-type operators similar to the ones modelling bimolecular chemical reactions in gas mixtures. In this model, then, binary interactions between individuals may lead to both an exchange of wealth and to a transfer to another subpopulation as a consequence of the \textit{same} binary interaction, i.e. the microscopic state identifying the wealth and the one related to the label, that denotes the subpopulation, change \textit{simultaneously}.

Another topic that has gained much interest in recent years also in kinetic theory, is the modelling of the spread of an epidemic: in this respect, compartmental Boltzmann models allowing the passage of individuals from an epidemiological compartment to another have been proposed, essentially of SIR type, where
a susceptible individual could become infected and then removed because of healing or death \cite{dimarco2020PRE, dellamarca2022NHM, dellamarca2022preprint}. A different kinetic description of infectious diseases consists in modelling interactions among different types of human cells, including the immune cells \cite{machadoramos2019JMB, dellamarca2022MMAS}. For example, in some epidemic models, susceptible (carrying a vanishing viral load~$v$) and infected individuals (with $v>0$) interact exchanging the quantity $v$, and as soon as the susceptible individual's viral load becomes positive because of the binary interaction, then he/she becomes infected \cite{dellamarca2022NHM,dellamarca2022preprint}.
In these works, the authors, similarly to \cite{loy2021KRM}, start from a microscopic description in which each agent is characterized by the microscopic quantity $v$ and by the label~$x$ denoting, in this case, the compartment. The microscopic dynamics is then described through discrete in time stochastic processes in which the new microscopic physical quantity and label are modelled through Markov-type jump processes governed by suitable transition probabilities \cite{dellamarca2022NHM,dellamarca2022preprint}. As a consequence, the kinetic model implementing the prescribed microscopic dynamics is a kinetic equation with an operator in which the kernel is related to the transition probability.

This formulation of the Boltzmann equation is well known also in kinetic theory for a single gas. Indeed, besides the classical Boltzmann operator in which the kernel has a proper cross section taking into account the intermolecular potential (depending on the relative speed and on the impact angles) \cite{cercignani1988BOOK}, other different forms have been used in the literature. The most common one is the so--called Waldmann representation \cite{waldamnn1958HdP}, showing in the kernel the probability distribution of the collision process transforming pre-collision velocities $(v,w)$ into the post-collision ones $(v', w')$, and the Boltzmann integral over the unit sphere is replaced by integrals over post-collision velocity variables. An analogous scattering kernel formulation replaces Waldmann kernel by its integral over the velocity of the partner molecule \cite{spiga1985PhysA}. The equivalence between these kinetic equations has been proved in \cite{boffi1990PhysA} for microscopic interactions that conserve the average and the energy. In the literature of kinetic equations for socio-economic sciences, the authors in \cite{loy2021KRM} show the equivalence between the collision-like Boltzmann equation and Markovian jump-processes described by transition probabilities that  can be related to the Waldmann (probabilistic) representation of the Boltzmann equation.

%Stating the microscopic model in this way allows to establish links to measurable quantities/observable mechanisms in some socio-economic phenomena such as international trade \cite{bisi2021PhTB} and epidemic disease spreading \cite{dellamarca2022NHM}.

In this paper we will present kinetic models for socio-economic problems in which agents have a multiple microscopic state, starting from a microscopic stochastic process ruled by transition probabilities that allow to describe the simultaneous change of all the microscopic variables and by a microscopic state-dependent interaction frequency. In the case in which the agent state is given by a microscopic quantity $v$ and by a label $x$ denoting the subgroup, we will show that this approach has some advantages with respect to the classical collision model \cite{bisi2021PhTB}. Indeed, as explained also in \cite{bisi2021PhTB}, the construction of the gain term of the Boltzmann operators requires the invertibility of the collision process, that obviously holds in gas-dynamics (because of conservation of total momentum and energy in each collision), but not in human interactions, that are also influenced by non-deterministic (random) effects. This invertibility property is not needed in the operator with a transition probability in the kernel, because each single interaction has its own probability, not related to the reverse process. Moreover, realistic situations with non constant interaction probabilities are easier to manage in the stochastic Boltzmann formulation, therefore this approach could have many applications in kinetic modelling of social sciences. For these reasons in this paper we present a formal and organic treatment of kinetic equations involving a microscopic stochastic process that simultaneously changes several internal states of the interacting agents (typically, their compartment and the value of their kinetic variable).
In more detail, the paper is organized as follows.

In Section 2 we formally derive the general form of a kinetic model implementing a microscopic dynamics in which each agent is characterized by a set $\zb\in \Omega \subset \mathbb{R}^d$ of microscopic states which may change simultaneously in each binary interaction, that is described by a transition probability and ruled by a microscopic state-dependent frequency. Even though the procedure is quite classical, stating the discrete in time stochastic process will be useful for writing the Nanbu-Babovski Algorithm for simulating the microscopic dynamics. Then we revise and establish the relation with some well-known models such as the collision-like Boltzmann equation and the kinetic equations describing transfers among different groups due to binary interactions. Finally, we explicitly derive the kinetic equation for a multi-agent system in which a binary interaction causes simultaneously a transfer and an exchange of a microscopic quantity.
In Section~3, mathematical properties of the Cauchy problem associated to our general Boltzmann equation are discussed, proving existence and uniqueness of a solution in suitable Wasserstein spaces. Then, the quasi-invariant limit commonly used to investigate socio-economic kinetic models is adapted to our general frame, allowing to derive suitable Fokker-Planck equations with additional terms taking into account transfers of agents.
Section 4 is devoted to the investigation of a specific kinetic model fitting into our general framework, describing international trade with possible transfers of individuals: evolution of number density and mean wealth of each country are computed from the kinetic model, the quasi-invariant limit is performed, and analogies and differences with respect to analogous models for a single population \cite{cordier2005JSP} are discussed, with particular reference to the Pareto index of steady distributions.
In Section 5 we show some numerical tests, simulating our kinetic equations by means of a Nanbu-Babovski Monte Carlo algorithm implementing the discrete in time stochastic process presented in Section 2: the evolution of distribution functions and of macroscopic quantities are commented on for varying parameters. Section 6 contains some concluding remarks and perspectives.

\section{Kinetic models for binary interactions processes}
In this section, we provide a formal derivation of kinetic equations implementing binary interactions among agents whose microscopic state is a vector, described by means of Markovian processes, where the interaction frequency depends on the microscopic state of the interacting agents. Then, we illustrate the relation to well-known kinetic models for binary interactions leading to exchange of physical quantities (the collision-like Boltzmann equation) and to label-switch processes, also named \textit{transfers}. We eventually present a general framework for describing, through kinetic equations, microscopic binary interaction processes leading to both exchanges of a physical quantity and label-switches.
\subsection{Kinetic models with transition probabilities}
Let us consider a large system of agents described by a microscopic state $\zb\in \Omega \subset \mathbb{R}^d$.  We shall suppose that the change of the microscopic state (of all its components simultaneously) is due to stochastic \textit{binary interactions}. A probabilistic description of such interactions may be given by means of \textit{transition probability} functions
\begin{equation}
	T(\zb'\vert \zb, \yb)>0, \qquad \tilde{T}(\yb'|\zb,\yb)>0 \qquad \forall\zb, \yb \in \Omega,\ t>0,
	\label{trans_prob}
\end{equation}
namely the \textit{conditional} probabilities that, given a binary interaction between an agent $\zb$ and an agent $\yb$, the first changes into $\zb'$ while the second into $\yb'$, respectively. Such a microscopic description may be assimilated to a \textit{Markov-type jump process}. In order for $T(\zb'\vert \zb, \yb), \tilde{T}(\yb'|\zb,\yb)$ to be conditional probability densities, they have to satisfy the following further property:
\begin{equation}\label{eq:norm.T}
\int_{\Omega} T(\zb'\vert \zb, \yb) \, d\zb'=1, \quad \int_{\Omega}\tilde{T}(\yb'|\zb,\yb) \, d\yb'=1 \qquad \forall \zb, \yb \in \Omega,\ t>0.
\end{equation}
The binary interactions may happen with a frequency $\lambda_{\zb\yb}$, namely the frequency of the binary interactions between two agents having microscopic states $\zb, \yb$ depends on the microscopic states themselves.
We remark that the two transition probabilities $T$ and $\tilde{T}$ are given in order to take into account for possible asymmetries in the binary interactions. We remark that the symmetry of the binary interactions is here expressed by
\begin{equation}\label{eq:symm}
T(\zb'\vert \zb, \yb)=\tilde{T}(\yb'\vert \zb, \yb),
\end{equation}
and $\lambda_{\zb\yb}=\lambda_{\yb\zb}$.
As classically done \cite{pareschi2013BOOK}, a kinetic description of the multi-agent system can be derived by introducing discrete in time stochastic processes. Let $\Zb_t,\Yb_t \in \Omega$ be random variables describing the microscopic state of two agents at time $t>0$. Let $f=f(\zb,t)$ be the probability density function associated to our multi-agent system, i.e. the probability density function of the random variable of a given agent $\Zb_t$, thus satisfying
\begin{equation}\label{eq:mass_cons}
\int_{\Omega} f(\zb,t) \, d\zb =1.
\end{equation}
During a sufficiently small time $\Delta{t}>0$ the agents may or may not change their state $\Zb_t,\Yb_t$ depending on whether a binary interaction takes place or not. We express this discrete-in-time random process as
\begin{equation}
\begin{aligned}
	  \Zb_{t+\Delta{t}} &= (1-\Theta)\Zb_t+\Theta \Zb_t', \\
	  \Yb_{t+\Delta{t}} &= (1-\Theta)\Yb_t+\Theta \Yb_t',
	\label{eq:micro.rules.gen}
	\end{aligned}
\end{equation}
where $\Zb_t', \, \Yb_t'$ are random variables describing the new microscopic state of $\Zb_t$ and $\Yb_t$ respectively after a binary interaction and having \textit{joint} probability density functions $g=g(\Zb_t'=\zb'; \Zb_t=\zb, \Yb_t=\yb), \tilde{g}=\tilde{g}(\Yb_t'=\yb';  \Zb_t=\zb, \Yb_t=\yb) $, while $\Theta \in\{0,\,1\}$ is a Bernoulli random variable, which we assume to be independent of all the other random variables appearing in~\eqref{eq:micro.rules.gen}, discriminating whether a binary interaction takes place ($\Theta=1$) or not ($\Theta=0$) during the time $\Delta{t}$. In particular, we set the probability to change the microscopic state
\begin{equation}
	\P(\Theta=1)=\lambda_{\Zb_t\Yb_t} \Delta{t},
	\label{eq:bernoulli}
\end{equation}
where $\lambda_{\Zb_t\Yb_t}$ is the interaction frequency between agents with microscopic states $\Zb_t$ and $\Yb_t$.
Notice that, for consistency, we need $\lambda_{\Zb_t\Yb_t} \Delta{t} \le 1$.

\noindent Let now $\phi=\phi(\zb)$ be an observable quantity defined on $\zb \in \Omega$. From~\eqref{eq:micro.rules.gen}-\eqref{eq:bernoulli}, together with the assumed independence of $\Theta$, we see that the mean variation rate of $\phi$ in the time interval $\Delta{t}$ satisfies
\begin{align*}
	&\frac{\ave{\phi(\Zb_{t+\Delta{t}})}-\ave{\phi(\Zb_t)}}{\Delta{t}}+\frac{\ave{\phi(\Yb_{t+\Delta{t}})}-\ave{\phi(\Yb_t)}}{\Delta{t}}= \\
	&=\frac{\ave{(1-\lambda_{\Zb_t\Yb_t}\Delta t) \phi(\Zb_t)}+\Delta t\ave{\lambda_{\Zb_t\Yb_t}\phi(\Zb_t')}-\ave{\phi(\Zb_t)}}{\Delta{t}}\\
	&+\frac{\ave{(1-\lambda_{\Zb_t\Yb_t}\Delta t) \phi(\Yb_t)}+\Delta t\ave{\lambda_{\Zb_t\Yb_t}\phi(\Yb_t')}-\ave{\phi(\Yb_t)}}{\Delta t}
\end{align*}
where $\ave{C_t}$ denotes the average of the random variable $C_t$ with respect to its probability density function. Whence, we deduce the instantaneous time variation of the average of $\phi$ in the limit $\Delta{t}\to 0^+$ as
\begin{equation}\label{eq:ave_4}
\frac{d}{dt}\ave{\phi(\Zb_t)}= \dfrac{1}{2}\Big(\ave{\lambda_{\Zb_t\Yb_t}\phi(\Zb_t')}+\ave{\lambda_{\Zb_t\Yb_t}\phi(\Yb_t')}-\ave{\lambda_{\Zb_t\Yb_t}\phi(\Zb_t)}-\ave{\lambda_{\Zb_t\Yb_t}\phi(\Yb_t)} \Big)
\end{equation}
where we used the fact that $\ave{\phi(\Zb_t)}=\ave{\phi(\Yb_t)}$ that implies $\ave{\phi(\Zb_t)}+\ave{\phi(\Yb_t)}=2\ave{\phi(\Zb_t)}$.

\noindent We now specify the gain terms as
$$
\begin{aligned}[b]
\ave{\lambda_{\Zb_t\Yb_t}\phi(\Zb_t')}=\int_{\Omega}\int_{\Omega^2}\phi(\zb')\lambda_{\zb\yb}g(\zb';\zb,\yb)  \, d\zb d\yb \, d\zb',\\
 \ave{\lambda_{\Zb_t\Yb_t}\phi(\Yb_t')}=\int_{\Omega}\int_{\Omega^2}\phi(\yb')\lambda_{\zb\yb}\tilde{g}(\yb';\zb,\yb)\, d\zb d\yb \, d\yb',
\end{aligned}
$$
where $g$ and $\tilde{g}$ are the joint probability density functions of $\Zb_t'$ and $\Yb_t'$, respectively, and of the samples of the random variables $\Zb_t=\zb, \Yb_t=\yb$ at time $t$. The probability density functions $g$ and $\tilde{g}$  are defined as
\begin{equation}\label{def:g_gt}
g(\zb';\zb,\yb)=T(\zb'|\zb,\yb) \, f_2(\zb,\yb,t)  \qquad \tilde{g}(\yb';\zb,\yb)=\tilde{T}(\yb'|\zb,\yb) \, f_2(\zb,\yb,t),  \,
\end{equation}
where $f_2(\zb,\yb,t)$ is the joint distribution of the couple $(\zb,\yb)$ at time $t$. As typically done in kinetic theory, we assume \textit{propagation of chaos}, i.e. $\zb$ and $\yb$ are independently distributed, which allows us to perform the factorization $f_2(\zb,\yb,t)=f(\zb,t) f(\yb,t)$, so that we can write
$$
g(\zb';\zb,\yb)=T(\zb'|\zb,\yb) \, f(\zb,t) f(\yb,t)  \qquad \tilde{g}(\yb';\zb,\yb)=\tilde{T}(\yb'|\zb,\yb) \, f(\zb,t) f(\yb,t).
$$
It is immediate to verify that $g$ and $\tilde{g}$ are probability density functions thanks to \eqref{eq:norm.T} and \eqref{eq:mass_cons}.
Analogously, the loss terms can be naturally written as
$$
\begin{aligned}[b]
\ave{\lambda_{\Zb_t\Yb_t}\phi(\Zb_t)}=\int_{\Omega^2}\phi(\zb)\lambda_{\zb\yb} f(\zb,t) f(\yb,t) \, d\zb d\yb ,\\
 \ave{\lambda_{\Zb_t\Yb_t}\phi(\Yb_t)}=\int_{\Omega^2}\phi(\yb)\lambda_{\zb\yb}f(\zb,t) f(\yb,t) \, d\zb d\yb.
\end{aligned}
$$
Therefore Eq. \eqref{eq:ave_4} can be stated as
\begin{equation}\label{eq:Boltz.T.fr}
\begin{aligned}
\frac{d}{dt} \int_{\Omega} f(\zb,t) \phi(\zb) \, d\zb &= \dfrac{1}{2} \int_{\Omega}\int_{\Omega^2}\lambda_{\zb\yb}T(\zb'|\zb,\yb)\left(\phi(\zb')  -\phi(\zb)\right)f(\zb,t) f(\yb,t) \, d\zb d\yb d\zb' \\
&+\dfrac{1}{2} \int_{\Omega}\int_{\Omega^2}\lambda_{\zb\yb}\tilde{T}(\yb'|\zb,\yb)\left(\phi(\yb')\,  -\phi(\yb)\right) f(\zb,t) f(\yb,t) \, d\zb d\yb  d\yb',
\end{aligned}
\end{equation}
where we have used \eqref{eq:norm.T} in order to write the loss terms.

In the following, we shall illustrate three meaningful examples of kinetic models describing binary interaction processes by means of transition probabilities: $i)$ binary interactions causing exchange of physical quantities; $ii)$ binary interactions leading to transfers of individuals; $iii)$ binary interactions leading to both exchange of physical quantities and transfers of individuals.

\subsection{Boltzmann-type description of classical binary interaction dynamics}\label{Boltz_sec}
Let us consider the case in which the microscopic state of the agent is a non-negative physical quantity $\zb=v \in \Omega=\mathbb{R}_+$. Extensions to negative and possibly also bounded microscopic
states are mostly a matter of technicalities.
In general, as classically done in kinetic theory \cite{pareschi2013BOOK}, if $v,w \in \mathbb{R}_+$ denote the pre-interaction states of any two interacting agents, their post-interaction states $v',w'$ will be given by general interaction rules in the form
\begin{equation}\label{eq:bin_rules_gen}
v'=I(v,w)+D(v,w)\eta, \qquad w'=\tilde{I}(v,w)+\tilde{D}(v,w) \eta_*
\end{equation}
where $\eta$ and $\eta_*$ are independent random variables satisfying $\ave{\eta}=\ave{\eta_*}=0, \ave{\eta^2}=\ave{\eta_*^2}=1$, namely with zero average and unitary variance.
It is known
that an aggregate description of the (sole) binary interaction dynamics inspired by the
principles of statistical mechanics can be obtained by introducing a probability density
function $f = f(v, t) \ge 0$ such that $f(v, t)dv$ gives the proportion of agents having
at time $t$ a microscopic state comprised between $v$ and $v + dv$. Such a probability density
function satisfies a Boltzmann-type kinetic equation, which in weak form reads
\begin{equation}\label{eq:Boltz}
\frac{d}{dt} \int_{\R_+} f(v,t) \phi(v) \, dv=\dfrac{\lambda}{2}\ave{\int_{\R_+^2} \big( \phi(v')+\phi(w')-\phi(v)-\phi(w)\big)f(v,t) f(w,t) \, dv dw}
\end{equation}
where $\lambda$ is the interaction frequency that we here assume to be independent of the microscopic states of the agents.

On the other hand, if we want to describe the binary interactions through transition probabilities \eqref{trans_prob} with $\zb'=v', \yb'=w'$, we have that \eqref{eq:Boltz.T.fr} can be rewritten as
\begin{equation}\label{eq:boltz.T}
\begin{aligned}
\frac{d}{dt} \int_{\R_+} f(v,t) \phi(v) \, dv &= \dfrac{\lambda}{2} \int_{\R_+^2}\left( \int_{\R_+} \phi(v')T(v'|v,w)\, dv' -\phi(v)\right)f(v,t) f(w,t) \, dv dw\\
&+\dfrac{\lambda}{2} \int_{\R_+^2}\left( \int_{\R_+} \phi(w')\tilde{T}(w'|v,w)\, dw' -\phi(w)\right) f(v,t) f(w,t) \, dv dw.
\end{aligned}
\end{equation}
  We can define the first two statistical moments of the distribution function $f$ as:
$$ M(t):=\int_{\R_+} v f(t,\,v)\,dv, \qquad  E(t):=\int_{\R_+} v^2f(t,\,v)\,dv, $$
that represent the average and energy, respectively. As done in \cite{loy2020CMS} in the symmetric case, in order to establish a relation between \eqref{eq:Boltz}-\eqref{eq:bin_rules_gen} and \eqref{eq:boltz.T}, we investigate primarily the trend of the statistical moments as prescribed by the two different models. % \eqref{eq:boltz.T} and \eqref{eq:Boltz}-\eqref{eq:bin_rules_gen}.
Setting $\phi(v)=v, v^2$ in~\eqref{eq:boltz.T} and \eqref{eq:Boltz}-\eqref{eq:bin_rules_gen} yields the evolution equations for $M, E$ for a system of agents obeying the microscopic dynamics expressed in terms of transition probabilities \eqref{trans_prob} or interaction rules \eqref{eq:bin_rules_gen}, respectively. By comparing the evolution equations of $M$ and $E$ prescribed by the two different kinetic models, we see that the evolution is the same if we choose
\begin{align*}
 I(v,\,w)=V_T(v,\,w), \quad D(v,\,w)=\sqrt{E_T(v,\,w)-V_T^2(v,\,w)}=:D_T(v,\,w), \\
\tilde{I}(v,\,w)=V_{\tilde{T}}(v,\,w), \quad \tilde{D}(v,\,w)=\sqrt{E_{\tilde{T}}(v,\,w)-V_{\tilde{T}}^2(v,\,w)}=:D_{\tilde{T}}(v,\,w)
\end{align*}
where
$$ V_T(v,\,w):=\int_{\R_+} v'T(v'\,\vert\,v,\,w)\,dv', \qquad
	E_T(v,\,w):=\int_{\R_+} v'^2T(v'\,\vert\,v,\,w)\,dv' $$
	and
	$$ V_{\tilde{T}}(v,\,w):=\int_{\R_+} w'\tilde{T}(w'\,\vert\,v,\,w)\,dw', \qquad
	E_{\tilde{T}}(v,\,w):=\int_{\R_+} w'^2\tilde{T}(w'\,\vert\,v,\,w)\,dw' $$
denote the mean and the energy, respectively, of $T$ and $\tilde{T}$ for a given pair $(v,\,w)\in{\R_+}\times{\R_+}$ of pre-interaction states, while $D_T(v,\,w)$ and $D_{\tilde{T}}(v,\,w)$ are the standard deviations of $T$ and $\tilde{T}$, respectively.
Therefore, if dealing with \eqref{eq:Boltz} we can consider the collisions
\begin{equation}
	v'=V_T(v,\,w)+D_T(v,\,w)\eta, \qquad w'=V_{\tilde{T}}(v,\,w)+D_{\tilde{T}}(v,\,w)\eta_\ast,
	\label{eq:binary}
\end{equation}
and this choice makes formulations \eqref{eq:Boltz} and \eqref{eq:boltz.T} equivalent \textit{at the macroscopic level} (at least for the mass, average and energy).
As highlighted in \cite{loy2020CMS}, in general, \eqref{eq:Boltz} with \eqref{eq:bin_rules_gen} and \eqref{eq:boltz.T} are not the same kinetic equation, although with the choice \eqref{eq:binary} they account for the same evolution of the first and second
statistical moments of $f$. Nevertheless, if in \eqref{eq:boltz.T} we take
\begin{equation}\label{eq:T_delta}
T(v'|v,w)=\delta\Big(v'-(V_T(v,w)+D_T(v,w)\eta)\Big), \qquad \tilde{T}(w'|v,w)=\delta\Big(w'-(V_{\tilde{T}}(v,w)+D_{\tilde{T}}(v,w)\eta_\ast)\Big)
\end{equation}
where in the right-hand side $\delta$ is the Dirac delta, then we can
formally show that \eqref{eq:boltz.T} becomes exactly \eqref{eq:Boltz}-\eqref{eq:bin_rules_gen}. Of course, in this case the right hand side \eqref{eq:boltz.T} is meant to be written in brackets $\ave{\cdot}$.

\subsection{Label switch process caused by binary interactions}
Let us now consider the case in which $\Omega=\cI_n=\lbrace 1,...,n \rbrace$ and the microscopic discrete variable $x\in\cI_n$ is regarded as a \textit{label}, that may denote the belonging of the agent to a certain group or subpopulation.
We assume that \textit{label switches}, i.e. migrations across subpopulations, can be caused by binary interactions between agents, causing a transfer of (potentially) both of them, but in such a way that the total mass of the agents in the system is conserved. We say that this process is formally a Markov-type one because the probability to switch from the current labels $x, y$ to new labels $x',y'$ does not depend on how the agents reached previously the labels $x,y$.
In particular we denote
\begin{equation}\label{eq:P}
P_{xy}^{x'y'}:= P(x',y'\vert x,y)
\end{equation}
the conditional probability density function of switching to the groups $x',y'$ given the pre-interaction labels $x,y$.
\begin{remark} \label{rem:discrete_2}
Since the variables $x,y$ are discrete, the mapping $(x',y')\mapsto P(x',y'\vert x,y)$ is a discrete probability measure. Consequently, we actually have
\begin{equation}\label{eq:P_cons}
\int_{\cI_n^2} P(x',y'\vert x,y)\,dx'dy'=\sum_{x',y' \in \cI_n}P(x',y'\vert x,y)=1.
\end{equation}
\end{remark}

If we introduce the probability density function $f=f(x,t)\geq 0$ of the agents with label $x$ at time $t$, its evolution can be modelled by a kinetic equation describing a Markov-type jump process:
\begin{equation}
	\partial_tf(x',t)=\lambda\left(\int_{\cI_n^3} P_{xy}^{x'y'}f(x,t)f(y,t)\,dxdydy'-f(x',t)\right),
	\label{eq:mark_proc_2}
\end{equation}
where $\lambda>0$ is the (constant) switch frequency. In weak form~\eqref{eq:mark_proc_2} reads
\begin{equation}
	\frac{d}{dt}\int_{\cI_n}\psi(x)f(x,t)\,dx=\lambda\int_{\cI_n^2}\int_{\cI_n^2}(\psi(x')-\psi(x))P_{xy}^{x'y'}f(x,t)f(y,t)\,dx' dy'\,dx dy,
	\label{eq:jump_proc_weak_2}
\end{equation}
where $\psi:\cI_n\to\R$ is an observable quantity (test function) defined on $\cI_n$.
Equation \eqref{eq:jump_proc_weak_2} can be derived by \eqref{eq:Boltz.T.fr} by setting $\zb=x$ and
\[
P_{xy}^{x'y'}=\dfrac{T(x'|x,y)+\tilde{T}(y'|x,y)}{2}.
\]
Since $x\in\cI_n$ is discrete, we may conveniently represent the distribution function $f$ as
\begin{equation}
	f(x,t)=\sum_{i=1}^{n}f_i(t)\delta(x-i),
	\label{eq:f.delta.label_switch_2}
\end{equation}
where $\delta(x-i)$ is the Dirac distribution centred in $x=i$ and $f_i=f_i(t)\geq 0$ is the probability that an agent is labelled by $x=i$ at time $t$.
In this way, we reconcile the weak form~\eqref{eq:jump_proc_weak_2} with the convention introduced in Remark~\ref{rem:discrete_2}, and~\eqref{eq:jump_proc_weak_2} actually becomes
\begin{equation}\label{eq:jump_proc_weak_3}
 \sum_{i=1}^{n}\psi(i)f_i'(t)=\lambda\sum_{i,l=1}^{n}\sum_{j,k=1}^{n}(\psi(i)-\psi(j))P_{jk}^{il}f_j(t)f_k(t).
 \end{equation}
 We have conservation of the total mass thanks to \eqref{eq:P_cons}, as we can verify setting $\psi=1$ in \eqref{eq:jump_proc_weak_2}. Using \eqref{eq:f.delta.label_switch_2} and, then, setting $\psi=1$ in \eqref{eq:jump_proc_weak_3}, this corresponds to
\begin{equation}\label{eq:f_lab_cons}
\sum_{i=1}^n f_i(t)=1.
\end{equation}
Choosing $\psi$ such that $\psi(s)=1$ for a certain $s\in\cI_n$ and $\psi(x)=0$ for all $x\in\cI_n\setminus\{s\}$ we get in particular
\begin{equation}
	f_s'=\lambda\left(\sum_{j,k,l=1}^{n}P_{jk}^{sl}f_jf_k-f_s\right), \qquad s=1,\,\dots,\,n,
	\label{eq:fi.jump_2}
\end{equation}
where we have used \eqref{eq:P_cons} and \eqref{eq:f_lab_cons}.
If we allow the interaction frequency $\lambda_{xy}$ to depend on the labels of the interacting agents, then the equation becomes
$$ \sum_{i=1}^{n}\psi(i)f_i'(t)=\sum_{i,l=1}^{n}\sum_{j,k=1}^{n}(\psi(i)-\psi(j))\lambda_{jk}P_{jk}^{il}f_j(t)f_k(t) $$
so that
\begin{equation}
	f_s'=\sum_{j,k,l=1}^{n}\left(\lambda_{jk}P_{jk}^{sl}f_jf_k-\lambda_{sk}P_{sk}^{jl}f_sf_k\right), \qquad s=1,\,\dots,\,n.
	\label{eq:fi.jump_3}
\end{equation}
In general we define
\begin{equation}\label{def:rate}
\beta_{ij}^{kl}:=\lambda_{ij}P_{ij}^{kl}
\end{equation}
the \textit{rate} of transfer for a couple from the subgroups $(i,j)$ to $(k,l)$ as a consequence of a binary interaction between two agents labelled $i$ and $j$.
\subsection{Interacting particles with label switch and exchange of physical quantities}\label{subsec:int_and_switch}
Let us now consider the case in which an agent is characterized by a physical quantity $v\in \R_+$ and by a label $x\in \cI_n$ that, again, denotes the belonging of the agent to a certain group. Hence, now, the microscopic state is $\zb=(x,v)\in\Omega=\cI_n\times \mathbb{R}_+$. The present framework allows to describe a situation in which an agent, as a consequence of a single binary interaction, changes (simultaneously) both the microscopic quantity $v$ and the label $x$.
Agents within the same group, i.e. with the same label, are assumed to be indistinguishable. We can take into account the possibility that the interactions among agents with the same label differ from those among agents with different labels. In general, if $(x,v),\, (y,w)\in\cI_n\times\R_+$ denote the pre-interaction states of any two interacting agents, their post-interaction quantities $v',\,w'$ will be given by \eqref{eq:bin_rules_gen}, where $I=I_{xy},\tilde{I}=\tilde{I}_{xy},D=D_{xy},\tilde{D}=\tilde{D}_{xy}$ may depend on $x,y$. Moreover, agents can also migrate to other subgroups $x',y'$ and this microscopic transfer process is described by the probability function \eqref{eq:P}.

We now want to derive a kinetic equation for the joint distribution function $f=f(x,v,t)\geq 0$, such that $f(x,v,t)dv$ gives the proportion of agents labelled by $x\in\cI_n$ and having microscopic state comprised between $v$ and $v+dv$ at time $t$. The discreteness of $x$ allows us to represent $f$ as \cite{loy2021KRM}
\begin{equation}\label{def:f}
f(x,v,t)=\sum_{i=1}^{n}f_i(v,t)\delta(x-i),
\end{equation}
where $f_i=f_i(v,t)\geq 0$ is the distribution function of the microscopic state $v$ of the agents with label $i$ and, in particular, $f_i(v,t)dv$ is the proportion of agents with label $i$ whose microscopic state is comprised between $v$ and $v+dv$ at time $t$.

\noindent Since both the interactions and the label switching conserve the total mass of the system, we may assume that $f(x,v,t)$ is a probability distribution, namely:
\begin{equation}
	\int_{\R_+}\int_{\cI_n} f(x,v,t)\,dx\,dv=\sum_{i=1}^{n}\int_{\R_+} f_i(v,t)\,dv=1 \qquad \forall\,t>0.
	\label{eq:f.prob}
\end{equation}
Notice, however, that the $f_i$'s are in general not probability density functions because their $v$-integral varies in time due to the label switching. We denote by
\begin{equation}
	\rho_i(t):=\int_{\R_+} f_i(v,t)\,dv
	\label{eq:rhoi}
\end{equation}
the mass of the group of agents with label $i$, thus $0\leq\rho_i(t)\leq 1$ and
$$ \sum_{i=1}^{n}\rho_i(t)=1 \qquad \forall\,t>0. $$
Let us also define the first statistical moment of $f_i$
\[
M_i(t):=\int_{\mathbb{R}_+} v\, f_i(v,t) \, dv
\]
so that the average of the $i-$th group is
\[
m_i(t):=\dfrac{M_i(t)}{\rho_i(t)}.
\]

The kinetic evolution equation for $f(x,v,t)$, expressed as in \eqref{def:f}, is given by \eqref{eq:Boltz.T.fr}, where now $\zb=(x,v)$,
which has to hold for every $\phi=\phi(x,v):\cI_n\times\R_+\to\R$.
Hence the evolution equation for $f$ is
\begin{align}
	\begin{aligned}[b]
	\frac{d}{dt}\sum_{i=1}^n&\int_{\R_+}\phi(i,v) f_i(v,t)\,dv\\ &=\dfrac{1}{2}\int_{\R_+^2}\sum_{i,j,k=1}^{n}\lambda_{jk}\int_{\R_+} T((i,v')|(j,v),(k,w))\big(\phi(i,v')-\phi(j,v)\big)f_j(v,t)f_k(w,t) \, dv dw dv'\\
	&+\dfrac{1}{2}\int_{\R_+^2}\sum_{l,j,k=1}^{n}\lambda_{jk}\int_{\R_+}\tilde{T}((l,w')|(j,v),(k,w))\big(\phi(l,w')-\phi(k,w)\big)f_j(v,t)f_k(w,t) \, dv dw dw'.
	\end{aligned}
\end{align}
Choosing $\phi(x,v)=\psi(x)\varphi(v)$ with $\psi$ such that $\psi(s)=1$ for a certain $s\in\cI_n$ and $\psi(x)=0$ for all $x\in\cI_n \setminus\{s\}$, we finally obtain the following system of equations for the subgroup distributions~$f_s$
\begin{align}
	\begin{aligned}[b]
	\frac{d}{dt}\int_{\R_+}\varphi(v)f_s(v,t)\,dv &=\\
	&=\dfrac{1}{2}\int_{\R_+^2}\sum_{j,k,i=1}^{n}\int_{\R_+}\Big(\lambda_{jk}\varphi(v') T((s,v')|(j,v),(k,w))f_j(v,t)\\
	&-\lambda_{sk}\varphi(v)T((i,v')|(s,v),(k,w))f_s(v,t)\Big)f_k(w,t) \, dv dw dv'\\
	&+\dfrac{1}{2}\int_{\R_+^2}\sum_{j,k,l=1}^{n}\int_{\R_+}\Big(\lambda_{jk}\varphi(w') \tilde{T}((s,w')|(j,v),(k,w))f_k(w,t)\\
	&-\lambda_{js}\varphi(w)\tilde{T}((l,w')|(j,v),(s,w))f_s(w,t)\Big)f_j(v,t)dv  dw dw'.
	\end{aligned}
	\label{eq:boltz.fi_22}
\end{align}
In particular, in order to implement the microscopic process \eqref{eq:bin_rules_gen}-\eqref{eq:P}, we choose
\begin{equation}\label{eq:T_lab.bin}
\begin{aligned}[b]
T((x',v')|(x,v),(y,w))&=\ave{\int _{\cI_n}P_{xy}^{x'y'}(v,w)\delta(v'-(I_{xy}(v,w)+D_{xy}(v,w)\eta)) dy'}, \\
\tilde{T}((y',w')|(x,v),(y,w))&=\ave{\int _{\cI_n}P_{xy}^{x'y'}(v,w)\delta(w'-(\tilde{I}_{xy}(v,w)+\tilde{D}_{xy}(v,w)\eta_\ast)) dx'},
\end{aligned}
\end{equation}
where we remark that $P_{xy}^{x'y'}=P_{xy}^{x'y'}(v,w)$ may depend on the microscopic physical quantities of the interacting agents.
Considering \eqref{eq:T_lab.bin}, \eqref{eq:boltz.fi_22} becomes
\begin{align}
	\begin{aligned}[b]
	\frac{d}{dt}\int_{\R_+}\varphi(v)f_s(v,t)\,dv &=\\
	&=\dfrac{1}{2}\ave{\int_{\R_+^2}\sum_{j,k,l=1}^{n}\Big(\beta_{jk}^{sl}\varphi(I_{xy}(v,w)+D_{xy}(v,w)\eta) f_j(v,t)\Big.\\
	\Big.&\phantom{\dfrac{1}{2}\int_{\R_+^2}\sum_{j,k,l=1}^{n}\beta_{jk}^{sl}\varphi(I_{xy}(v,w)}-\beta_{sk}^{jl}\varphi(v)f_s(v,t)\Big)f_k(w,t) \, dv dw }\\
	&+\dfrac{1}{2}\ave{\int_{\R_+^2}\sum_{j,k,i=1}^{n}\Big(\beta_{jk}^{is}\varphi(\tilde{I}_{xy}(v,w)+\tilde{D}_{xy}(v,w)\eta_\ast) f_k(w,t)\Big.\\
	\Big.&\phantom{\dfrac{1}{2}\int_{\R_+^2}\sum_{j,k,l=1}^{n}\beta_{jk}^{sl}\varphi(I_{xy}(v,w)}-\beta_{js}^{ik}\varphi(w)f_s(w,t)\Big)f_j(v,t)dv  dw },
	\end{aligned}
	\label{eq:boltz.fi_23}
\end{align}
where we have used \eqref{def:rate}. We remark that now $\beta_{ij}^{kl}$ may depend on the microscopic variables $v$ and $w$ through $P_{ij}^{kl}$. Moreover, as done in \cite{loy2021EJAM} in a linear case, also the interaction frequency may depend on the microscopic state $v$.

We remark that in this case where we consider both exchanges \eqref{eq:bin_rules_gen} (with $I,\tilde{I},D$ and $\tilde{D}$ that may depend on the labels) and transfers \eqref{eq:P}, asymmetric binary interactions arise quite commonly, even if the two processes \eqref{eq:bin_rules_gen} and \eqref{eq:P}, separately, are symmetric. This is mainly due to the fact that the microscopic rule \eqref{eq:bin_rules_gen} depends on the label of the agents. Indeed, if we consider a transfer $(i,j) \rightarrow (k,l)$, the reverse transfer $(k,l) \rightarrow (i,j)$ may occur with a different probability and with a different interaction law, losing thus the reversibility of the process, usually assumed in classical Boltzmann descriptions.
Specifically, in gas mixtures the reversibility is guaranteed by conservations of momentum and total energy, and the post--collision velocities may be uniquely determined in terms of the pre--collision velocities and of the impact angles~\cite{chapman1970BOOK, kogan1969BOOK}, even in presence of chemical reactions \cite{rossani1999PhysA}. The break of symmetry between the direct and the reverse collision is known to occur in presence of inelastic collisions (for instance in granular media~\cite{bobylev2000JSP}) causing a decay in time of the kinetic energy of the system. For interactions involving human beings the kinetic approach is much more complicated (even under simplistic assumptions), and exchanges of goods and transfers among different compartments may be non--symmetric. Just to give an example, in socio-economic problems the fraction of the own wealth that each agent is willing to give to the others may depend also on the proper amount of wealth \cite{bisi2017UMI}, and moreover a transfer from a poor country to a rich one might be much more probable than the reverse transfer \cite{bisi2021PhTB}.
This is why the general approach with generally different transition probabilities $T$ and $\tilde{T}$ provides a useful tool for a correct description of this kind of processes. Moreover, it allows to build more easily exchange and transfer operators for a generic number $n$ of subpopulations. Indeed, the usual way of extending Boltzmann theory to a set of $n>1$ constituents consists in building up a set of $n$ Boltzmann equations, each one for the distribution function of the $i$-th constituent, with $i=1, \dots, n$ \cite{Gio1999, rossani1999PhysA, DB2008CMS}. On the other hand, our transition probability approach includes the label of the individual compartment into the set of microscopic states characterizing the individual, dealing thus with only one kinetic equation, that could be separated into $n$ different equations only when one needs to compute the pertinent moments of each compartment by choosing the appropriate test function as done in \eqref{eq:boltz.fi_22}. This turns out to be a great advantage even from the computational point of view. As we will see in numerical tests shown in Section 5, in the present approach it is also straightforward to consider transition rates $\beta_{ij}^{kl}$ explicitly dependent on the microscopic states, both through the binary interaction frequency and /or through the transfer probability $P_{ij}^{kl}$, investigating thus more realistic cases with respect to classical kinetic descriptions that, for the sake of simplicity, assume constant interaction probabilities in the kernel of the Boltzmann operators. Furthermore, this approach allows to include the stochastic contributions $\eta$ and $\eta_\ast$ in the binary interaction rules, as invertibility is not required as in the construction of the operators.

\section{Formal study of the kinetic equation with transition probabilities}
In this section we intend to revise and illustrate some analytical tools that are useful for formally studying equation \eqref{eq:Boltz.T.fr}. After briefly stating some results on the existence and uniqueness of the solution, we consider the quasi-invariant limit in various regimes of equation \eqref{eq:boltz.fi_23} involving label switching and exchange of physical quantities.
\subsection{Basic theory of kinetic models with transition probabilities in Wasserstein spaces}
 The strong form of \eqref{eq:Boltz.T.fr} coupled with an initial condition $f_0(\zb)$ defines the following Cauchy problem
\begin{equation}\label{eq:Cauchy}
\begin{cases}
&\dfrac{\partial}{\partial t}  f(\zb,t)  = Q^+(f,f)  -f(\zb,t)\displaystyle\int_{\Omega}\lambda_{\zb\yb}  f(\yb,t) \, d\yb, \qquad t>0, \quad \zb \in \Omega,\\
&f(0,\zb)=f_0(\zb), \qquad \zb \in \Omega,
\end{cases}
\end{equation}
where
\begin{equation}
Q^+(f,f)= \dfrac{1}{2}\int_{\Omega^2}\lambda_{\pr{\zb}\pr{\yb}} T(\zb|\pr{\zb},\pr{\yb})f(\pr{\zb},t) f(\pr{\yb},t) \, d\pr{\zb} d\pr{\yb}+\dfrac{1}{2}\int_{\Omega^2}\lambda_{\pr{\zb}\pr{\yb}} \tilde{T}(\zb|\pr{\yb},\pr{\zb})f(\pr{\zb},t) f(\pr{\yb},t) \, d\pr{\zb} d\pr{\yb},
\end{equation}
where $\pr{\zb},\pr{\yb}$ are the pre-interaction states, with the compatibility condition $\int_{\Omega}f_0(\zb) \, d\zb=1$ as \eqref{eq:mass_cons} holds true. We remark that everything could be written for a generic mass $\rho>0$.
Let us now define
$$\bar{\lambda}:=\int_{\Omega}\lambda_{\zb,\yb}  f(\yb,t) \, d\yb
$$
that we assume to be constant throughout the whole text (this assumption includes the case of a constant $\lambda_{\zb\yb}$). If we multiply both sides of the equation by $e^{\bar{\lambda}t}$ and we integrate in time we get
\begin{equation}\label{eq:sol}
\begin{aligned}[b]
f(\zb,t)=e^{-\bar{\lambda}t}f_0(\zb)+\int_0^t e^{\bar{\lambda}(s-t)}\left[\dfrac{1}{2}\int_{\Omega^2}\right.&\left.\lambda_{\pr{\zb},\pr{\yb}} T(\zb|\pr{\zb},\pr{\yb})f(\pr{\zb},s) f(\pr{\yb},s) \, d\pr{\zb} d\pr{\yb}	\right.\\
&\left.+\dfrac{1}{2}\int_{\Omega^2}\lambda_{\pr{\zb},\pr{\yb}} \tilde{T}(\zb|\pr{\yb},\pr{\zb})f(\pr{\zb},s) f(\pr{\yb},s) \, d\pr{\zb} d\pr{\yb}\right] ds,
\end{aligned}
\end{equation}
where we have used \eqref{eq:mass_cons}.
Let us define $(\Omega,d)$ a polish space. Analogously to what has been done in \cite{freguglia2017CMS}, we see that an appropriate space in which \eqref{eq:sol} can be studied is $X:=\mathcal{C}([0,\bar{t}];\mathcal{M}_+(\Omega))$, where $\bar{t}>0$ is a final time and $\mathcal{M}_+(\Omega)$ is the space of positive measures on $\Omega$ having unitary mass. Therefore $f \in X$ is a continuous mapping as a function of time over $[0,\bar{t}]$ and it is a positive measure satisfying \eqref{eq:mass_cons} as a function of the microscopic state $\zb \in \Omega$. In particular, $X$ is a complete state with the distance
\[
\sup_{t\in [0,\bar{t}]} W_1(f(t,\zb),g(t,\yb))
\]
where
\begin{equation}\label{def:W1}
W_1(f(t,\zb),g(t,\yb)) = \inf_{\underline{\mu}\in \Gamma(f,g)} \int_{\Omega^2} d(\zb,\yb) \underline{\mu}(t,\zb,\yb) \, d\zb d\yb
\end{equation}
is the 1-Wasserstein distance between $f(t, \cdot)$ and $g(t,\cdot )\in \mathcal{M}_+(\Omega)$, being $\Gamma(f,g)$ the space of the probability density functions defined on $\Omega^2$ having marginals $f$ and $g$.

As done in \cite{freguglia2017CMS}, we shall always assume that the transition probabilities $T$ and $\tilde{T}$ satisfy the following Lipschitz continuity property.
\begin{Assumption}\label{assumptionLipT}
Let $T(\zb|\pr{\zb},\pr{\yb}),\tilde{T}(\zb|\pr{\zb},\pr{\yb}) \in \mathscr{P}(\Omega)$ for all $\pr{\zb},\pr{\yb} \in \Omega$,
where $\mathscr{P}(\Omega)$ is the space of probability measures on $\Omega$. We assume that there exists
$Lip(T) > 0$, such that
\[
W_1(T(\cdot|\pr{\zb},\pr{\yb}), T(\cdot|\pr{\zb}_*,\pr{\yb}_*)) \le \textrm{Lip}(T) (|\pr{\zb}-\pr{\zb}_* |+|\pr{\yb}-\pr{\yb}_*|)
\]
for all $\pr{\zb},\pr{\yb},\pr{\zb}_*,\pr{\yb}_* \in \Omega$ and that the same holds for $\tilde{T}$.
\end{Assumption}

The following result holds.
\begin{theorem}\label{Theorem_basic}
Let $f_0 \in \mathcal{M}_+(\Omega)$ and
let us assume that $T$, $\tilde{T}$ satisfy assumption \ref{assumptionLipT}, that $W_1(f_0,T),W_1(f_0,\tilde{T})<\infty$, that $\bar{\lambda}$ is constant and that $\lambda_{\zb\yb}$ is lower and upper bounded, i.e. $\exists\, \tilde{\lambda}^m,\tilde{\lambda}^M>0$ such that $0<\tilde{\lambda}^m<\lambda_{\zb\yb} <\tilde{\lambda}^M<\infty$.
Then, there exists a unique $f\in X$ which solves \eqref{eq:Cauchy}, and \eqref{eq:Cauchy} exhibits continuous dependence on the initial data.

\noindent Moreover, if $\tilde{\lambda}^M\textrm{Lip}(T ),\tilde{\lambda}^M\textrm{Lip}(\tilde{T})<\dfrac{1}{2}$, then \eqref{eq:Cauchy} admits a unique equilibrium distribution $f_{\infty}$, which is a probability measure on $\Omega$ and which is also globally attractive, i.e.
\[
\lim_{t \rightarrow \infty} W_1(f(t; \cdot); f_{\infty})= 0
\]
for every solution $f$ to \eqref{eq:Cauchy}.
\end{theorem}

\begin{proof}
The proof follows the same steps as done in \cite{freguglia2017CMS} Appendix A, where the authors prove the results in the case of a bounded $\Omega$ and, thus, use the dual form of the 1-Wasserstein distance due to the Rubinstein-Kantorovitch Theorem \cite{ambrosio}. In the present case, as $\Omega$ is arbitrary, we can use in the proof the definition \eqref{def:W1} recalling the hypothesis $W_1(f_0,T),W_1(f_0,\tilde{T})<\infty$ and $\tilde{\lambda}^m<\lambda_{\zb\yb}<\tilde{\lambda}^M$.
\end{proof}

Moreover, the following Theorem holds in the case of label switching and exchange of physical quantities.
\begin{theorem}
Let the transition probability distributions have the form
\[
T(x,v|(\pr{x},\pr{v}),(\pr{y},\pr{w}))=\sum_{i=1}^{n}T_{i}(v|(\pr{x},\pr{v}),(\pr{y},\pr{w}))\delta(x-i)
\]
where $T_i$ satisfies
\[
|T_i(v|(\pr{x},\pr{v}),(\pr{y},\pr{w}))| <Lip(T_i)(|\pr{x}-\pr{y}|+|\pr{v}-\pr{w}|)
\]
and let the analogous property hold for $\tilde{T}$ and $\tilde{T}_i$.
Let moreover
\[
f_0(x,v)=\sum_{i=1}^{n} f^{(0)}_i(v)\,\delta(x-i)
\]
be a prescribed kinetic distribution function at time $t = 0$ over the space of microscopic states $(x,v)\in \cI_n\times\R_+$ such that $f_i \ge 0$, $\sum_{i=1}^n\int_{\R_+}f_i(v,t)\, dv=1, \, \forall t>0$.
Then the unique solution to \eqref{eq:Cauchy} is of the form
$$
f(x,v,t) = \sum_{i=1}^n f_i(v,t)\, \delta (x-i)
$$
(analogous to \eqref{eq:f.delta.label_switch_2}), with coefficients $f_i(v,t)$ given by \eqref{eq:fi.jump_3} along with the initial conditions
$f_i(v,0) = f^{(0)}_i(v)$. In addition, it depends continuously on the initial datum as stated by
Theorem~\ref{Theorem_basic}.
\end{theorem}

\subsection{Quasi-invariant limit}
One of the most interesting issues in the study of kinetic models is the characterisation of the stationary
distributions arising asymptotically for $t\rightarrow \infty$, which depict the emergent behaviour of the system.
The jump process model~\eqref{eq:Boltz.T.fr} hardly allows one to investigate in detail the trend to equilibrium and the profile of the stationary distributions, and the explicit expression of the steady state $f_\infty$ can be inferred only in particular cases \cite{delitala2014KRM,puppo2017KRM,freguglia2017CMS}.

%Conversely, the collisional model~\eqref{eq:Boltz}-\eqref{eq:binary} offers several analytical tools, which often permit to recover explicitly accurate approximations of $f_\infty$ by means of suitable asymptotic procedures.
%The basic idea of such procedures is to approximate the integro-differential equation~\eqref{eq:Boltz} with partial differential equations, more amenable to analytical investigations, at least in some regimes of the parameters of the microscopic collisions. %Indeed, identifying the equilibrium distributions directly from~\eqref{eq:Boltz} by setting the right-hand side to zero remains a quite difficult, and often unsolvable, task.

It is widely known that the classic collisional Boltzmann equation~\eqref{eq:Boltz}-\eqref{eq:bin_rules_gen} offers several analytical tools, which often permit to explicitly recover accurate approximations of $f_\infty$ by means of suitable asymptotic procedures.
The basic idea of such procedures is to approximate an integro-differential Boltzmann equation with an appropriate partial differential equation, more amenable to analytical investigations, at least in some regimes of the parameters of the microscopic interactions.
A prominent framework  in which this type of asymptotic analysis is successfully applied to~\eqref{eq:Boltz} is that of the \textit{quasi-invariant interactions}. This concept was first introduced in the kinetic literature on multi-agent systems in~\cite{cordier2005JSP,toscani2006CMS} as a reminiscence of the \textit{grazing collisions} studied in the classical kinetic theory, see~\cite{villani1998ARMA}. This corresponds to introducing a small parameter $\epsilon$ such that the microscopic interaction rule can be written as
\begin{equation}\label{eq:qinv_micro}
v' \approx v+\mathcal{O}(\epsilon)
\end{equation}
and analyzing the dynamics on a longer time scale, setting a new time variable
\begin{equation}\label{eq:tau}
\tau:=\epsilon t
\end{equation}
in order to compensate for the smallness of the interactions.

In this spirit, in \cite{loy2020CMS}, that concerns the investigation of a parallelism between the model \eqref{eq:Boltz}-\eqref{eq:bin_rules_gen} and \eqref{eq:boltz.T}, the authors propose a way to translate the concept of quasi-invariancy, typically used in the context of collision-like Boltzmann equations \eqref{eq:Boltz}-\eqref{eq:bin_rules_gen}, to the language of transition probabilities. The idea is the following. Let $\pr{\Zb},\,\Zb\in\R_+$ be the random variables representing the pre- and post-interaction states, respectively, of an agent, and $\pr{\Zb_\ast}\in\R_+$ the one representing the pre-interaction state of the other agent involved in the interaction. In the probabilistic description via the transition probabilities, we say that interactions are quasi-invariant if, given $0<\epsilon\ll 1$,
\begin{equation}
	\P(\abs{\Zb-\pr{\Zb}}>\epsilon\,\vert\,\pr{\Zb},\,\pr{\Zb_\ast})\leq\epsilon;
	\label{eq:quasi-inv_int}
\end{equation}
in other words, if the post-interaction state is, in probability, close to the pre-interaction state, so that the interactions produce a small transfer of microscopic state between the interacting agents.

In the present framework, in order to have a \textit{quasi-invariant transition probability}, we can introduce rescaled transition probabilities defined by the following transform
\begin{equation}\label{eq:T_eps}
T_{\epsilon}(\zb|\pr{\zb},\pr{\yb})=\mathcal{F}_\epsilon [T](\zb|\pr{\zb},\pr{\yb}), \qquad \tilde{T}_{\epsilon}(\yb|\pr{\zb},\pr{\yb})=\mathcal{F}_\epsilon[\tilde{T}](\yb|\pr{\zb},\pr{\yb})
\end{equation}
where
\[
\mathcal{F}_\epsilon [T] : \mathscr{P}_1(\R_+) \longmapsto \mathscr{P}_1(\R_+)
\]
is a family of operators (for $\epsilon>0$) defined on
the space of the probability measures defined on $\mathbb{R}_+$. We require that $\mathcal{F}_\epsilon$ satisfies the following three properties:
\begin{itemize}
\item[$F1$] $\mathcal{F}_1$ is the identity;
\item[$F2$]
\begin{equation}\label{prop1}
\lim_{\epsilon \rightarrow 0} W_1(\mathcal{F}_\epsilon [T],\delta(\zb'-\zb))=0, \qquad \lim_{\epsilon \rightarrow 0} W_1(\mathcal{F}_\epsilon [\tilde{T}],\delta(\zb'-\zb))=0
\end{equation}
\item[$F3$]
\begin{equation}\label{prop2}
\lim_{\epsilon \rightarrow 0} W_1(\dfrac{\mathcal{F}_\epsilon [T]}{\epsilon},T)=0, \qquad \lim_{\epsilon \rightarrow 0} W_1(\dfrac{\mathcal{F}_\epsilon [\tilde{T}]}{\epsilon},\tilde{T})=0
\end{equation}
\end{itemize}
meaning that $\epsilon=1$ corresponds to the basic regime ($F1$), that for small values of $\epsilon$ the microscopic state tends not to change ($F2$), and that on the long time scale \eqref{eq:tau} the dynamics is ruled by $T$ and $\tilde{T}$ ($F3$) \cite{loy2020CMS}.
An example of properly rescaled transition probabilities is
\begin{equation}\label{eq:T_eps_delta}
T_{\epsilon}(\zb'|\zb,\yb)=(1-\epsilon)\delta(\zb'-\zb)+\epsilon T(\zb'|\zb,\yb), \qquad \tilde{T}_{\epsilon}(\yb'|\zb,\yb)=(1-\epsilon)\delta(\yb'-\yb)+\epsilon \tilde{T}(\yb'|\zb,\yb)
\end{equation}
as introduced in \cite{loy2020CMS}, satisfying the properties $F1$, $F2$, $F3$.

In the following we shall investigate the quasi invariant limit in the three examples illustrated in the previous section.

\subsubsection{Boltzmann-type description of classical binary interaction dynamics}
The quasi-invariant limit procedure is classically applied to the collisional Boltzmann equation with microscopic interaction rules \eqref{eq:Boltz}-\eqref{eq:bin_rules_gen} (or \eqref{eq:Boltz}-\eqref{eq:binary}). Let us introduce a small parameter $0<\epsilon\ll 1$, a time scale \eqref{eq:tau} and a corresponding probability density function $f^{\epsilon}(\tau,v)=f(\tau/\epsilon,v)$. For what we have said in Section 2, as shown in \cite{loy2020CMS}, we have that the quasi-invariant microscopic rule for having the same evolution of the average and the energy of both $f$ and $f^{\epsilon}$ on the $t$- and $\tau$-time scale, respectively, is
\[
\begin{aligned}[b]
& v'=V_{T_{\epsilon}}(v,w)+D_{T_\epsilon}(v,w)\eta, \\
\end{aligned}
\]
and analogously for the microscopic rule for $w'$.
If we consider the quasi-invariant transition probability \eqref{eq:T_eps_delta}, we have that $ V_{T_{\epsilon}}(v,w)=v+\epsilon(I(v,w)-v)$ and $D_{T_\epsilon}(v,w)= \sqrt{\epsilon(1-\epsilon)(I(v,w)-v)^2+\epsilon D(v,w)^2}$ so that the quasi-invariant microscopic rules are
\begin{equation}\label{eq:qinv_bin}
\begin{aligned}[b]
& v'=v+\epsilon(I(v,w)-v)+\sqrt{\epsilon(1-\epsilon)(I(v,w)-v)^2+\epsilon D(v,w)^2} \, \eta), \\
& w'=w+\epsilon(\tilde{I}(v,w)-w)+\sqrt{\epsilon(1-\epsilon)(\tilde{I}(v,w)-w)^2+\epsilon \tilde{D}(v,w)^2} \, \eta_\ast).
\end{aligned}
\end{equation}
Therefore, in terms of transition probabilities, in order to recover the same evolution of the first two moments in the two models \eqref{eq:Boltz}-\eqref{eq:bin_rules_gen} and \eqref{eq:boltz.T}, as shown in the previous section, we must consider $T_\epsilon$ having average $\epsilon(I(v,w)-v)$ and variance $\epsilon(1-\epsilon)(I(v,w)-v)^2+\epsilon D(v,w)^2$ and analogously for $\tilde{T}_{\epsilon}$.
In particular, in order for \eqref{eq:Boltz}-\eqref{eq:bin_rules_gen} and \eqref{eq:boltz.T} to be the same model an appropriate choice is \eqref{eq:T_eps_delta} with $\zb=v$ or
\begin{equation}\label{eq:delta}
\begin{aligned}[b]
&T_\epsilon(v'|v,w)=\delta\Big(v'-(v+\epsilon(I(v,w)-v)+\displaystyle\sqrt{\epsilon(1-\epsilon)(I(v,w)-v)^2+\epsilon D(v,w)^2} \, \eta)\Big), \\
&\tilde{T}_\epsilon(w'|v,w)=\delta\Big(w'-(w+\epsilon(\tilde{I}(v,w)-w)+\sqrt{\epsilon(1-\epsilon)(\tilde{I}(v,w)-w)^2+\epsilon \tilde{D}(v,w)^2} \, \eta_\ast) \Big).
\end{aligned}
\end{equation}
We remark that for $\epsilon=1$ we recover $T$ and $\tilde{T}$, while for $\epsilon \rightarrow 0$ we have that
\[
\begin{aligned}[b]
W_1(T_\epsilon(\cdot|v,w),\delta(\cdot-v))& \le \int_\mathbb{R_+^2}|v'-w'| T_\epsilon(v'|v,w)\delta(w'-v) \, dv' dw' \\
&=|\epsilon(I(v,w)-v)+\sqrt{\epsilon(1-\epsilon)(I-v)^2+\epsilon D^2} \, \eta| \longrightarrow_{\epsilon\rightarrow 0} 0
\end{aligned}
\]
and it can also be easily verified that $F3$ holds true for $T_\epsilon$ defined by \eqref{eq:delta},
and analogously for $\tilde{T}_{\epsilon}$.
If we consider for a moment the symmetric case \eqref{eq:symm}, for simplicity, we have that plugging~\eqref{eq:tau} into~\eqref{eq:boltz.T} and considering \eqref{eq:T_eps} satisfying \eqref{prop2}, letting $\epsilon$ to $0^+$ yields
\begin{equation}
	\partial_\tau f=\int_{\R_+}\int_{\R_+}T(v\,\vert\,\pr{v},\,\pr{w})f(\tau,\,\pr{v})f(\tau,\,\pr{w})\,d\pr{v}\,d\pr{w}-f,
	\label{eq:boltz.g-T}
\end{equation}
which is structurally identical to the very general equation~\eqref{eq:boltz.T} and does not give any further information. Therefore, in spite of the quasi-invariant structure of the interactions, it is in principle not easier to extract from~\eqref{eq:boltz.g-T} any more detailed information about the asymptotic trends.

\noindent Instead, if we consider \eqref{eq:boltz.T} with \eqref{eq:delta} and \eqref{eq:tau}, then we can perform the quasi-invariant limit, as in \cite{loy2020CMS}, and we obtain letting $\epsilon$ to $0^+$
\begin{align}
	\begin{aligned}[b]
		\partial_\tau f &= \frac{1}{2}\partial^2_v\left\{\left[\int_{\R_+}\left(\left(V_T(v,\,w)-v\right)^2
			+D_T^2(v,\,w)\right)f(\tau,\,w)\,dw\right]f\right\} \\
		&\phantom{=} -\partial_v\left[\left(\int_{\R_+} V_T(v,\,w)f(\tau,\,w)\,dw-v\right)f\right],
	\end{aligned}
	\label{eq:FP_0}
\end{align}
where $V_T,D_T$ are the average and variance of $T$ defined in Section \ref{Boltz_sec}.

\subsubsection{Label switch process caused by binary interactions}
Let us now consider the transfer process described by the kinetic equation \eqref{eq:fi.jump_3}. In order to write a quasi-invariant regime, we can express the fact that, given a collision, individuals have a small probability of jumping, i.e.
\begin{equation}\label{eq:P_eps}
P_{\epsilon_{xy}}^{x'y'}=\epsilon P_{xy}^{x'y'} \quad \text{if} \quad (x,y)\neq (x',y'), \quad   P_{\epsilon_{xy}}^{x'y'}=1-\epsilon \quad \text{if} \quad (x,y)= (x',y').
\end{equation}
Then, considering a long time scale \eqref{eq:tau} we have that \eqref{eq:fi.jump_3} is
\begin{equation}
	\dfrac{d f_s^\epsilon}{d\tau}=\dfrac{1}{\epsilon}\sum_{j,k,l=1, \\ j\neq s, k\neq l}^{n}\left(P_{\epsilon_{jk}}^{sl}\lambda_{{jk}}f_j^\epsilon-P_{\epsilon_{sk}}^{jl}\lambda_{sk}f_s^\epsilon\right)f_k^\epsilon, \qquad s=1,\,\dots,\,n.
	\label{eq:fi.jump_5}
\end{equation}
Plugging \eqref{eq:P_eps} in \eqref{eq:fi.jump_5} we obtain
\[
\dfrac{d f_s^{\epsilon}}{d\tau}=\sum_{j,k,l=1}^{n}\left(\beta_{jk}^{sl}f_j^\epsilon f_k^\epsilon-\beta_{sk}^{jl}f_s^\epsilon f_k^\epsilon \right), \qquad s=1,\,\dots,\,n
\]
which is structurally identical to the very general equation \eqref{eq:fi.jump_3} with $\tau$ instead of $t$, meaning that it is the probability of jumping that rules the dynamics on the long time scale.
%\textbf{Remark} Considerare una small jump frequency
%\[
%\lambda_{jk}^{il}\rightarrow \lambda_{\epsilon_{jk}}^{il}=\epsilon \lambda_{jk}^{il}
%\]
%\`e equivalente a riscalare il tempo come \eqref{eq:tau}.
%For completeness, we report here, in the notation of this paper, the statement of the
%result proved in \cite{freguglia2017CMS}, concerning the form of the aymptotic equilibrium state.
%
%\begin{theorem}
%Let the transition probability distributions have the form
%\[
%T(x|\pr{x},\pr{y})=\sum_{i=1}^{n}T_{i}(\pr{x},\pr{y})\delta(x-i), \qquad \tilde{T}(x|\pr{x},\pr{y})=\sum_{i=1}^{n}\tilde{T}_{i}(\pr{x},\pr{y})\delta(x-i)
%\]
%where $T_i$ and $\tilde{T}_i$ are Lipschitz continuous. Let moreover
%\[
%f_0(x)=\sum_{i=1}^{n} f_i(t)\delta(x-i)
%\]
%be a prescribed kinetic distribution function at time $t = 0$ over the lattice of microscopic states $x\in \cI_n$ such that $f_i \ge 0$, $\sum_{i=1}^nf_i(t)=1, \, \forall t>0$
%Then the unique solution to \eqref{eq:Cauchy} is \eqref{eq:f.delta.label_switch_2} with coefficients $f_i(t)$ given by \eqref{eq:fi.jump_3} along with the initial conditions
%$f_i(0) = f_i^{(0)}$. In addition, it depends continuously on the initial datum as stated by
%Theorem \ref{Theorem_basic}.
%\end{theorem}

\subsubsection{Interacting particles with label switch and exchange of physical quantities}

Let us now consider the case in which the binary interactions lead both to a transfer and to an exchange of the physical quantity $v$. Therefore, for the two processes we shall consider a quasi-invariant regime given by \eqref{eq:qinv_micro}-\eqref{eq:P_eps}. In this case the transition probability \eqref{eq:T_lab.bin} may be rescaled as
\begin{equation}\label{eq:qinv_T}
T_{\epsilon}((x',v')|(x,v),(y,w))=P_{\epsilon_{xy}}^{x'y'}\delta\Big(v'-(v+\epsilon(I_{xy}-v)+\sqrt{\epsilon(1-\epsilon)(I_{xy}-v)^2+\epsilon D_{xy}^2} \, \eta) \Big)
\end{equation}
and analogously for $\tilde{T}_{\epsilon}$. Let us consider the symmetric case for simplicity of notation, bearing in mind that the asymmetric case can be treated analogously. Plugging the latter in \eqref{eq:boltz.fi_23} and considering the re-scaling \eqref{eq:tau} and reminding \eqref{eq:P_eps}, if we let $\epsilon \rightarrow 0^+$ we obtain
\begin{align}\label{eq:q.inv_tr_ex}
	\begin{aligned}[b]
	&\frac{d}{d\tau}\int_{\R_+}\varphi(v)f_s(\tau,v)\,dv =\\
	&=\int_{\R_+^2}\sum_{j,k,l=1, j\neq s \vee k\neq l}^{n}\varphi(v)\big(\beta_{jk}^{sl}f_j(\tau,v)-\beta_{sk}^{jl}f_s(\tau,v)\big)f_k(\tau,w) \, dv dw
	\end{aligned}
\end{align}
that means that the dynamics is ruled by the label switches and gives no further information. In case of symmetry, it is the symmetric form of \eqref{eq:boltz.fi_23}.

%If there is simmetry then the latter is
%\begin{equation}\label{eq:quasi_inv}
%	\begin{aligned}[b]
%	&\frac{d}{d\tau}\int_{\R_+}\varphi(v)f_s(v,t)\,dv =\\
%	&\int_{\R_+^2}\sum_{j,k,l=1, j\neq s \vee k\neq l}^{n}\varphi(v)\big(\beta_{jk}^{sl}f_j(\tau,v)-\beta_{sk}f_s(\tau,v)\big)f_k(\tau,w) \, dv dw
%	\end{aligned}
%\end{equation}
%that is the weak form of \eqref{eq:fi.jump_3}.

Let us now consider a different regime and, in particular, let us only consider a quasi-invariant exchange rule, i.e. \eqref{eq:qinv_micro}.
Let us then rescale \eqref{eq:boltz.fi_23} with \eqref{eq:tau} and let us consider the quasi-invariant transition probability
\begin{equation}\label{eq:trans_qinv.rules}
T_{\epsilon}((x',v')|(x,v),(y,w))=P_{{xy}}^{x'y'}\delta\Big(v'-(v+\epsilon(I_{xy}-v)+\sqrt{\epsilon(1-\epsilon)(I_{xy}-v)^2+\epsilon D_{xy}^2} \, \eta) \Big),
\end{equation}
i.e. the exchange of the physical quantity is actually quasi-invariant, whilst the label-switch process is not.
 We obtain
\begin{align}\label{eq:fs_qinv}
	\begin{aligned}[b]
	&\frac{d}{d\tau}\int_{\R_+}\varphi(v)f_s(\tau,v)\,dv =\\
	& = \ave{\dfrac{1}{\epsilon}\int_{\R_+^2}\sum_{j,k,l=1}^{n}\big(\beta_{jk}^{sl}\varphi(v+\epsilon(I_{jk}-v)+\sqrt{\epsilon(1-\epsilon)(I_{jk}-v)^2+\epsilon D_{jk}^2}\eta) f_j(\tau,v)\\
	&\phantom{=\dfrac{1}{\epsilon}\int_{\R_+^2}\sum_{j,k,l=1}^{n}\big(\beta_{jk}^{sl}\varphi(v+\epsilon(I_{jk}-v)+}-\beta_{sk}^{jl}\varphi(v)f_s(\tau,v)\big)f_k(\tau,w) \, dv dw}  .
	\end{aligned}
\end{align}
As $\epsilon$ is small, we can Taylor expand $$\begin{aligned}[b]
\ave{\varphi(v+\epsilon(I_{jk}(v,w)-v)+\sqrt{\epsilon(1-\epsilon)(I_{jk}-v)^2+\epsilon D_{jk}^2}\eta)}=\\
=\varphi(v)+\epsilon(I_{jk}(v,w)-v)\varphi'(v)+\dfrac{1}{2}\varphi''(v)\epsilon \left((1-\epsilon)(I_{jk}-v)^2+ D_{jk}^2\right)+\mathcal{O}(\epsilon^2)
\end{aligned}
$$
and plugging the latter in \eqref{eq:fs_qinv} we have that
%\begin{align*}
%	\begin{aligned}[b]
%	&\frac{d}{d\tau}\int_{\R_+}\varphi(v)f_s(\tau,v)\,dv =\\
%	&\ave{\dfrac{1}{2}\int_{\R_+^2}\sum_{j,k,l=1}^{n}\big(\beta_{jk}^{sl}\varphi'(v)(I_{jk}(v,w)-v)+(1-\epsilon)((I_{jk}-v)^2+ D_{jk}^2)\varphi''(v) f_j(\tau,v)\big)f_k(\tau,w) \, dv dw} \\
%	&+\ave{\dfrac{1}{2}\int_{\R_+^2}\sum_{i,j,k=1}^{n}\big(\beta_{jk}^{is}\varphi'(w)(\tilde{I}_{jk}(v,w)-w)+(1-\epsilon)((\tilde{I}_{jk}-w)^2+\epsilon \tilde{D}_{jk}^2)\varphi''(w) f_k(\tau,w)\big)f_j(\tau,v)dv  dw}\\
%	&+\dfrac{1}{2\epsilon}\int_{\R_+^2}\sum_{j,k,l=1, j\neq s}^{n}\varphi(v)\big(\beta_{jk}^{sl}f_j(\tau,v)-\beta_{sk}^{jl}f_s(\tau,v)\big)f_k(\tau,w) dv dw\\
%	&+\dfrac{1}{2\epsilon}\int_{\R_+^2}\sum_{i,j,k=1, k\neq s}^{n}\varphi(w)\big(\beta_{jk}^{is}f_k(\tau,w)-\beta_{js}^{ik}f_s(\tau,w)\big)f_j(\tau,v) dv dw
%	\end{aligned}
%\end{align*}
%and, therefore
\begin{align*}
	\begin{aligned}[b]
	&\frac{d}{d\tau}\int_{\R_+}\varphi(v)f_s(\tau,v)\,dv =\\
	&=\int_{\R_+^2}\sum_{j,k,l=1}^{n}\beta_{jk}^{sl}\big(\varphi'(v)(I_{jk}-v)+\dfrac{1}{2}((1-\epsilon)(I_{jk}-v)^2+ D_{jk}^2)\varphi''(v)\big)f_j(\tau,v)f_k(\tau,w) \, dv dw \\
	&+\dfrac{1}{\epsilon}\int_{\R_+}\sum_{j,k,l=1, j\neq s}^{n}\varphi(v)\big(\beta_{jk}^{sl}f_j(\tau,v)-\beta_{sk}^{jl}f_s(\tau,v)\big)\rho_k(\tau) \,dv
	\end{aligned}
\end{align*}
where we remind that $\rho_k$ is the mass of the $k$-th population.
Let us now consider an expansion for the probability density function of the whole population $f$
\begin{equation}\label{eq:expansion}
f(\tau,x,v)=f^{(0)}(\tau,x,v)+\epsilon f^{(1)}(\tau,x,v) +\mathcal{O}(\epsilon^2)
\end{equation}
where the zero-th and first order moments satisfy
\begin{equation}\label{Hilbert.cons}
\begin{aligned}[b]
\int_{\cI_n\times \mathbb{R}_+}f  dx dv=\rho^{(0)}:=\int_{\cI_n\times \mathbb{R}_+}f^{(0)}  \, dx dv, \qquad \rho^{(1)}:=\int_{\cI_n\times \mathbb{R}_+}f^{(1)}  \, dx dv=0,\\
\int_{\cI_n\times \mathbb{R}_+}f v dx dv=M^{(0)}:=\int_{\cI_n\times \mathbb{R}_+}f^{(0)} v \, dx dv, \qquad M^{(1)}:=\int_{\cI_n\times \mathbb{R}_+}f^{(1)} v \, dx dv=0.
\end{aligned}
\end{equation}
For each $f_i$ this translates into
\[
f_i(\tau,v)=f_i^{(0)}(\tau,v)+\epsilon f_i^{(1)}(\tau,v)+\mathcal{O}(\epsilon^2)
\]
and \eqref{Hilbert.cons} translates to
\begin{equation}\label{eq:Hilbert}
\rho^{(0)}=\sum_{i=1}^n \rho_i^{(0)}=\int_{\cI_n\times \mathbb{R}_+}f dx dv, \qquad \sum_{i=1}^n \rho_i^{(1)}=0.
\end{equation}
Comparing equal orders of $\epsilon$ and supposing symmetry, we obtain
\begin{equation}\label{eq:f_0_1}
\begin{aligned}[b]
&\dfrac{1}{\epsilon}\int_{\R_+}\sum_{j,k,l=1, j\neq s}^{n}\varphi(v)\big(\beta_{jk}^{sl}f_j^{(0)}(\tau,v)-\beta_{sk}^{jl}f_s^{(0)}(\tau,v)\big)\rho_k^{(0)}(\tau) dv=0, \\
	\end{aligned}
\end{equation}
so that
\begin{equation}\label{eq:f_0_2}
\begin{aligned}[b]
&f_s^{(0)}(\tau,v) =\dfrac{\sum_{j,k,l=1, j\neq s}^{n}\beta_{jk}^{sl}f_j^{(0)}(\tau,v)\rho_k^{(0)}(\tau)}{\sum_{j,k,l=1, j\neq s}^{n}\beta_{sk}^{jl}\rho_k^{(0)}(\tau)}, \\
	\end{aligned}
\end{equation}
while, at the first order
\begin{equation}\label{eq:f_1}
	\begin{aligned}[b]
	&\frac{d}{d\tau}\int_{\R_+}\varphi(v)f_s^{(0)}(\tau,v)\,dv =\\
	&=\int_{\R_+^2}\sum_{j,k,l=1}^{n}\beta_{jk}^{sl}\big(\varphi'(v)(I_{jk}(v,w)-v)+\frac12 ((I_{jk}(v,w)-v)^2+ D_{jk}^2)\varphi''(v)\big) f_j^{(0)}(\tau,v) f_k^{(0)}(\tau,w) \, dv dw \\
&+\int_{\R_+}\sum_{j,k,l=1, j\neq s}^{n}\varphi(v)\big(\beta_{jk}^{sl}f_j^{(0)}(\tau,v)-\beta_{sk}^{jl}f_s^{(0)}(\tau,v)\big)\rho_k^{(1)}(\tau) dv \\
&+\int_{\R_+}\sum_{j,k,l=1, j\neq s}^{n}\varphi(v)\big(\beta_{jk}^{sl}f_j^{(1)}(\tau,v)-\beta_{sk}^{jl}f_s^{(1)}(\tau,v)\big)\rho_k^{(0)}(\tau) dv,
	\end{aligned}
\end{equation}
from which we can obtain, resorting to the strong form thanks to integration by parts, a Fokker-Planck-type equation with a reaction term for each $f_s^{(0)}$ that is
\begin{equation}\label{eq:f_1_1}
	\begin{aligned}[b]
	\partial_\tau f_s^{(0)}(\tau,v) &=\\
	&-\partial_v\sum_{j,k,l=1}^{n}\beta_{jk}^{sl}\int_{\R_+}(I_{jk}(v,w)-v)f_k^{(0)}(\tau,w) \, dw f_j^{(0)}(\tau,v)\\
	 &+\partial^2_{vv}\dfrac{1}{2}\sum_{j,k,l=1}^{n}\beta_{jk}^{sl}\int_{\R_+}((I_{jk}(v,w)-v)^2+ D_{jk}^2)f_k^{(0)}(\tau,w) \, dw f_j^{(0)}(\tau,v)\\
	&+\sum_{j,k,l=1, j\neq s}^{n}\big(\beta_{jk}^{sl}f_j^{(0)}(\tau,v)-\beta_{sk}^{jl}f_s^{(0)}(\tau,v)\big)\rho_k^{(1)}(\tau)\\
	&+\sum_{j,k,l=1, j\neq s}^{n}\big(\beta_{jk}^{sl}f_j^{(1)}(\tau,v)-\beta_{sk}^{jl}f_s^{(1)}(\tau,v)\big)\rho_k^{(0)}(\tau).
	\end{aligned}
\end{equation}
The latter reaction terms also involve the first order corrections $f_i^{(1)}, \, i=1,...,n$.
In order to find univocally the solutions $f_s^{(0)},f_s^{(1)}$ to \eqref{eq:f_0_2} and \eqref{eq:f_1_1} satisfying \eqref{eq:Hilbert}, we need a number of conditions (to be looked for example in conserved quantities) that is equal to the number of degrees of freedom.

\section{Kinetic model for international trade allowing transfer of individuals}
In this section, we are going to rephrase with the current framework the kinetic model for international trade allowing  transfer of individuals investigated in \cite{bisi2021PhTB}, where the author presents a model of interacting individuals divided into two subpopulations and allowed, by means of binary interactions, to exchange wealth and to migrate to the other subgroup. Here, then, we have that $n=2$, the physical quantity $v\in\R_+$ is the wealth, while the label $x\in \cI_2$ denotes the subgroup. Note that, in  the present work, we are only going to consider binary interactions giving rise to both exchanges of the wealth and transfer simultaneously.
\subsection{From the microscopic to the macroscopic model}\label{sec:model}
For what concerns the exchange of the physical quantity $v$, we are going to consider simple linear microscopic rules \eqref{eq:bin_rules_gen}
with
%\begin{equation}
%	V'_t:=(1-\omega_{X_t})V_t+\omega_{Y_t}W_t+\sigma_{xy}V_t\eta, \qquad W'_t:=\omega_{X_t}V_t+(1-\omega_{Y_t})W_t+\sigma_{xy}W_t\eta_\ast
%	\label{eq:micro.rule_2}
%\end{equation}
%In the notation of the present work, this implies
\begin{equation}\label{eq:micro.rule_2}
I_{xy}(v,w)=(1-\omega_{x})v+\omega_{y}w, \qquad D_{xy}=\zeta_{xy}v
\end{equation}
where $\omega_i \in [0,1], \, i\in\cI_2$ and we are dealing with symmetric interactions.
We consider possible transfers given by
\begin{equation}\label{eq:transfer_mod}
\begin{aligned}[b]
\textrm{(a)} \qquad 1 + 1 \rightarrow 1 + 2, \qquad \textrm{(b)} \qquad 2 + 2 \rightarrow 1 + 2,\\
\textrm{(c)} \qquad 1 + 2\rightarrow 1 + 1, \qquad \textrm{(d)} \qquad 1 + 2 \rightarrow 2 + 2,
\end{aligned}
\end{equation}
therefore only one of the two interacting agents moves to the other subgroup.
The latter implies that the only non-vanishing values of $P_{ij}^{kl}$ correspond to the 4-plets
\begin{equation}\label{eq:P_mod}
 (i,j,k,l)  \in \lbrace (1,1,1,2), (1,1,2,1), (2,2,1,2),(2,2,2,1),(1,2,1,1),(1,2,2,2),(2,1,1,1),(2,1,2,2)\rbrace.
\end{equation}

%\begin{equation}\label{eq:micro_omega}
%I_{xy}(v,w)=(1-\omega_x)v+\omega_y w, \quad \tilde{I}_{xy}(v,w)=(1-\omega_y)w+\omega_x v, \qquad D_{xy}=\zeta_{xy}v,\quad \tilde{D}_{xy}=\zeta_{xy}w,
%\end{equation}

The kinetic equation describing this microscopic dynamics is \eqref{eq:boltz.fi_22} with \eqref{eq:T_lab.bin}, where $P_{ij}^{kl}$ is defined according to \eqref{eq:P_mod} and the microscopic exchange dynamics by \eqref{eq:micro.rule_2}.
The evolution of the mass of each subpopulation is given by setting $\varphi=1$ in \eqref{eq:boltz.fi_22}, with the prescribed dynamics, for $s=1,2$ and it results in
\begin{align}\label{eq:masse}
\begin{aligned}[b]
\partial_t \rho_1=(\beta_{22}^{12}\rho_2+\beta_{12}^{11}\rho_1)\rho_2-(\beta_{11}^{12}\rho_1+\beta_{12}^{22}\rho_2)\rho_1\\
\partial_t \rho_2=(\beta_{11}^{12}\rho_1+\beta_{12}^{22}\rho_2)\rho_1-(\beta_{22}^{12}\rho_2+\beta_{12}^{11}\rho_1)\rho_2
\end{aligned}
\end{align}
while setting $\varphi=v$ in \eqref{eq:boltz.fi_22} gives the evolution of the first moments for $s=1,2$
\begin{align}\label{eq:momenti}
\begin{aligned}[b]
\partial_t M_1=\left(\beta_{12}^{11}\rho_1+\beta_{22}^{12}\rho_2\right)M_2
-\left(\beta_{12}^{22}\rho_2+\beta_{11}^{12}\rho_1\right)M_1,\\
\partial_t M_2=\left(\beta_{12}^{22}\rho_2+\beta_{11}^{12}\rho_1\right)M_1-\left(\beta_{12}^{11}\rho_1+\beta_{22}^{12}\rho_2\right)M_2.
\end{aligned}
\end{align}
As a consequence, the averages of the wealth of population 1 and 2 evolve as
\begin{align}\label{medie}
\begin{aligned}[b]
\partial_t m_1=\dfrac{\rho_2}{\rho_1}\left( \beta_{22}^{12}\rho_2+\beta_{12}^{11}\rho_1\right)(m_2-m_1),\\
\partial_t m_2=\dfrac{\rho_1}{\rho_2}\left( \beta_{11}^{12}\rho_1+\beta_{12}^{22}\rho_2\right)(m_1-m_2).
\end{aligned}
\end{align}
We remark that we have assumed symmetry in the interaction rates, i.e.
\[
\beta_{ii}^{12}=\beta_{ii}^{21}, \quad \beta_{12}^{ii}=\beta_{21}^{ii}, \quad \forall i=1,2.
\]
We observe that the total mass and average
\[
\bar{\rho}:=\rho_1+\rho_2=1, \qquad \bar{M}:=M_1+M_2
\]
are conserved in time.
Regarding the stationary states of the masses, we have that
\[
\rho_2^{\infty}=\alpha \rho_1^{\infty}
\]
where
\[
\alpha=\dfrac{-(\beta_{12}^{11}-\beta_{12}^{22})+\sqrt{(\beta_{12}^{11}-\beta_{12}^{22})^2+4\beta_{11}^{12}\beta_{22}^{12}}}{2\beta_{22}^{12}},
\]
and, taking into account that the sum of the two densities is constant, we have that
\[
\rho_1^{\infty}=\dfrac{\bar{\rho}}{1+\alpha}, \qquad \rho_2^{\infty}=\dfrac{\alpha\bar{\rho}}{1+\alpha}.
\]
Therefore
\[
\rho_1^{\infty}=\rho_2^{\infty}=\dfrac{\bar{\rho}}{2} \qquad \textrm{if and only if} \qquad \alpha=1.
\]
Bearing in mind that $\beta_{ij}^{kl}=P_{ij}^{kl}\lambda_{ij}$, the latter condition is satisfied if $P_{12}^{22}=P_{12}^{11}=0.5$ and $\lambda_{11}=\lambda_{22}$, which means that the probability for interacting agents with different labels and going to the same subgroup is the same, and the frequency of interaction among agents of the same subgroup is the same for all subgroups.
%\item if $\lambda_{11}=0$ .
%\end{itemize}
To this regard, we observe that, as the stationary state only depends on $\alpha$, then there may be a switch in the population size ($(\rho_2^{\infty}-\rho_1^{\infty})(\rho_2(0)-\rho_1(0))<0$) if
\begin{equation}\label{eq:mass_switch}
(\rho_2(0)-\rho_1(0))(\alpha-1)<0.
\end{equation}
For what concerns the average, the sufficient and necessary condition to be met at the stationary state is
\begin{equation}\label{eq:medie_staz}
m_1^{\infty}=m_2^{\infty}=:m^{\infty}
\end{equation}
for every choice of the parameters.
The latter implies that $M_2^{\infty}>M_1^{\infty}$ if and only if $\alpha>1$. Moreover, because of conservation of mass and total momentum, we have that
\begin{equation}\label{eq:media_staz}
m^{\infty}=\rho_1(0)m_1(0)+\rho_2(0)m_2(0)
\end{equation}
that implies that the final average wealth is closer to the initial average wealth of the subgroup that was more populated at $t=0$.

\subsection{Quasi-invariant limit}

If we assume, for simplicity of notation, that the probability of transfer towards the $i$-th subgroup is independent of the countries of the interacting agents, i.e.
\begin{equation}\label{assumption}
\beta_1^2 := \beta_{11}^{12} = \beta_{12}^{22}, \qquad \beta_2^1 := \beta_{22}^{12} = \beta_{12}^{11}
\end{equation}
we have that
\begin{equation}
\rho_2^\infty=\alpha \rho_1^\infty, \qquad M_2^\infty=\alpha M_1^\infty,
\end{equation}
where $\alpha=\beta_1^2/\beta_2^1$,
and, as $\bar{M}$ and $\bar{\rho}=1$ are conserved quantities, we have that
\[
\rho_1^{\infty}=\dfrac{\bar{\rho}}{1+\alpha}, \qquad M_1^{\infty}=\dfrac{\bar{M}}{1+\alpha}.
\]
Let us now consider the quasi-invariant regime defined by the transition probability
\begin{equation}\label{eq:trans_qinv.rules_noncons}
T_{\epsilon}((x',v')|(x,v),(y,w))=P_{{xy}}^{x'y'}\delta\Big(v'-(v+\epsilon(I_{xy}(v,w)-v)+\sqrt{\epsilon D_{xy}(v,w)^2}\eta) \Big).
\end{equation}
The latter, even if it satisfies the requirements F1, F2, F3 prescribed in Section 3.2, differently from \eqref{eq:trans_qinv.rules}, does not guarantee the same evolution of the energy in the quasi-invariant regime.
By expanding the distribution functions in powers of $\epsilon$, imposing that the globally invariant quantities (zero-th and first order moments) remain unexpanded, we get that the constraints \eqref{eq:Hilbert} are
\begin{equation}\label{eq:Hilbert_2}
\rho_1^{(1)}+\rho_2^{(1)}=0, \qquad M_1^{(1)}+M_2^{(1)}= 0.
\end{equation}
Therefore, we have two degrees of freedom and we can determine the values of $\rho_1^{(1)}, M_1^{(1)}$ as we have two conserved quantities.
%We always have that $\int_{\cI_n\times \mathbb{R}_+}f dx dv=\rho_1+\rho_2=1$.  Therefore, we consider a microscopic rule such that there is another conserved quantity $M_1+M_2=\bar{M}$.
%In particular, we shall consider \eqref{eq:bin_rules_gen}-\eqref{eq:micro.rule_2}.
Following the same procedure as before, we find that \eqref{eq:f_0_2} now is
\begin{equation}\label{eq:f_0}
f_1^{(0)}=\dfrac{1}{\alpha}f_2^{(0)}\,.
\end{equation}
Then we have that $\rho_2^{(0)}=\alpha \rho_1^{(0)}$ and $M_2^{(0)}=\alpha M_1^{(0)}$, i.e. $\rho_i^{(0)}=\rho_i^\infty, M_i^{(0)}=M_i^\infty, \, i=1,2$, which means that the masses and averages of order zero correspond to the equilibrium ones.
Therefore
\begin{equation}\label{eq:rho_M_0}
\rho_1^{(0)}=\dfrac{\bar{\rho}}{1+\alpha}, \qquad M_1^{(0)}=\dfrac{\bar{M}}{1+\alpha}\,.
\end{equation}
At the first order \eqref{eq:f_1_1} for $s=1$ (for $s=2$ an analogous result applies) specialises into
\begin{equation}\label{eq:FP}
\begin{aligned}[b]
\partial_\tau f_1^{(0)} &=
-\partial_v\Big(\beta_{1}^{2}(\omega_1 M_1^{(0)}-\omega_2v\rho_1^{(0)})f_1^{(0)}+\beta_{1}^{2}(\omega_2M_2^{(0)}-\omega_2v\rho_2^{(0)})f_1^{(0)}\Big.\\
\Big.&\phantom{ =
-\partial_v (\beta_{1}^{2}}+\beta_{1}^{2}(\omega_1M_1^{(0)}-\omega_1v\rho_1^{(0)})f_1^{(0)}+\beta_{2}^{1}(\omega_2M_2^{(0)}-\omega_1v\rho_2^{(0)})f_1^{(0)}\Big)\\
&\phantom{=}+\partial^2_{v^2}\dfrac{1}{2}\Big(\beta_{1}^{2}\left[\zeta_{12}^2v^2\right]f_1^{(0)}\rho_1+\beta_{1}^{2}\left[\zeta_{22}^2v^2\right]f_1^{(0)}\rho_2\Big.\\
\Big. &\phantom{ =
-\partial_v (\beta_{1}^{2}}+\beta_{1}^{2}\left[\zeta_{11}^2v^2\right]f_1^{(0)}\rho_1^{(0)}+\beta_{2}^{1}\left[\zeta_{12}^2v^2\right]f_1^{(0)}\rho_2^{(0)}\Big) \\
	&\phantom{=}+\left(f_2^{(1)}\beta_2^{1}-f_1^{(1)}\beta_1^2\right)\bar{\rho}
\end{aligned}
\end{equation}
where use of \eqref{eq:f_0} has been made.
Integrating \eqref{eq:FP} over $\mathbb{R}_+$ and \eqref{eq:FP}  multiplied by $v$ over $\mathbb{R}_+$, along with the conditions $f_1(0)=0$ and $\lim_{v\rightarrow +\infty} f_1(v)=0$, and remembering \eqref{eq:Hilbert_2}, we discover that both $\rho_1^{(1)}=\rho_2^{(1)}=0$ and
$M_1^{(1)}=M_2^{(1)}=0$.
This implies the fact that both the masses and the averages of $f_1$ and $f_2$ are at the equilibrium even at $\mathcal{O}(\epsilon)$ accuracy.
Using relations \eqref{eq:f_0}-\eqref{eq:rho_M_0} in \eqref{eq:FP},  we obtain
\begin{equation}\label{eq:FP_simpl}
\begin{aligned}[b]
\partial_\tau f_1^{(0)} &=
-\dfrac{\beta_1^2}{1+\alpha}\partial_v\left[\Big((\omega_1 \bar{M}-\omega_2v\bar{\rho})+\alpha(\omega_2\bar{M}-\omega_2v\bar{\rho})+(\omega_1\bar{M}-\omega_1v\bar{\rho})+(\omega_2\bar{M}-\omega_1v\bar{\rho})\Big)f_1^{(0)}\right]\\
&\phantom{=}+\dfrac{\zeta^2\beta_1^2\bar{\rho}}{2(1+\alpha)}\partial^2_{v^2}\Big(v^2(3+\alpha)f_1^{(0)}\Big) \\
	&\phantom{=}+\left(f_2^{(1)}\beta_2^{1}-f_1^{(1)}\beta_1^2\right)\bar{\rho}
\end{aligned}
\end{equation}
where we have also assumed that the stochastic fluctuations are the same in each kind of interaction, i.e.
\[
\zeta_{ij}=\zeta, \qquad \forall i,j \in \cI_2.
\]
Equation \eqref{eq:FP_simpl} is a Fokker-Planck equation with reaction term, where the advection-diffusion part (first and second line) only involves $f_1^{(0)}$ and its (known) mass and average \eqref{eq:rho_M_0}, while the reaction term (third line in \eqref{eq:FP_simpl}) only depends on $f_1^{(1)},f_2^{(1)}$.
As $\rho_1^{(1)}=\rho_2^{(1)}=0$ and $M_1^{(1)}=M_2^{(1)}=0$, the reaction term does not influence the mass and average of $f_1^{(0)}$. It is therefore reasonable to look for the stationary solution to the Fokker-Planck equation without reaction term, i.e.
\begin{equation}\label{eq:FP_simpl_2}
\begin{aligned}[b]
\partial_\tau \tilde{f}_1^{(0)} &=
-\dfrac{\beta_1^2}{1+\alpha}\partial_v\left[\Big((\omega_1 \bar{M}-\omega_2v\bar{\rho})+\alpha(\omega_2\bar{M}-\omega_2v\bar{\rho})+(\omega_1\bar{M}-\omega_1v\bar{\rho})+(\omega_2\bar{M}-\omega_1v\bar{\rho})\Big)\tilde{f}_1^{(0)}\right]\\
&\phantom{=}+\dfrac{\zeta^2\beta_1^2\bar{\rho}}{2(1+\alpha)}\partial^2_{v^2}\Big(v^2(3+\alpha)\tilde{f}_1^{(0)}\Big)
\end{aligned}
\end{equation}
that is \eqref{eq:FP_simpl} where we neglect the third term on the right hand side, as this does not contribute to a variation of mass and average of $f_1^{(0)}$. Therefore we obtain
\begin{equation}\label{stato_staz}
\tilde{f}_1^{(0)}=\dfrac{\bar{\rho}}{1+\alpha}v^{-2\left(1+\dfrac{\gamma}{2 }\right)}\exp^{-\dfrac{\bar{M}}{\bar{\rho}}\dfrac{\gamma}{v}}, \qquad \gamma=\dfrac{B}{D}, \quad B=2\omega_1+\omega_2(1+\alpha), \quad D=\zeta^2 \dfrac{3+\alpha}{2},
\end{equation}
The mass and average can be verified to be $\bar{\rho}/(1+\alpha) $ and $\dfrac{\bar{M}}{1+\alpha}$ respectively, while the energy is $\dfrac{\bar{M}}{(1+\alpha)(\gamma-1)}$.
Moreover, we can determine the \textit{Pareto index} of the first population that is (approximated by)
\begin{equation}
PI_1=\gamma+1=\dfrac{2\omega_1+\omega_2(1+\alpha)}{\frac{\zeta^2}{2}(3+\alpha)}+1
\end{equation}
that depends on the trading propensity of both populations $\omega_1,\omega_2$, on the ratio $\alpha$ that involves the rates $\beta_i^j$ and on the stochasticity $\zeta^2$. Since, according to (\ref{eq:f_0}), $\tilde{f}_2^{(0)}= \alpha\, \tilde{f}_1^{(0)}$, both populations have the same (approximate) Pareto index.
%As \eqref{eq:f_0} holds, then
%the Pareto index of the second population can be approximated by
%\[
%PI_2=\alpha PI_1=\alpha\dfrac{2\omega_1+\omega_2(1+\alpha)}{\frac{\zeta^2}{2}(3+\alpha)}+\alpha.
%\]

\noindent We remark that considering only one population corresponds to setting $\omega_1=\omega_2=\omega$ and $\beta_1^2=\beta_2^1$.
If $\beta_2^1=\beta_1^2$, then there is no reaction term in \eqref{eq:FP_simpl} so that the stationary state \eqref{stato_staz} is exact and $\alpha=1$ that implies $f_1^{(0)}=f_2^{(0)}$.
Moreover, the Pareto index (now exact) is
\[
\dfrac{2\omega}{\zeta^2}+1,
\]
that coincides with the one commonly obtained from a kinetic model for a single population \cite{cordier2005JSP}.
Note that, even keeping $\alpha \not= 1$, in the case $\omega_1 = \omega_2$ the Pareto index of each group coincides with that relevant to a single population.

%In the general case with $\omega_1 \not= \omega_2$, we note that the Pareto index $PI_1$ is a monotone function, ranging from $\dfrac{2\omega_1 + \omega_2}{3/2 \zeta^2} + 1$ (for $\alpha \rightarrow 0$) to $2\dfrac{\omega_2}{\zeta^2} +1$ (for $\alpha \rightarrow \infty$.
%
%Mah.... Rileggendo, forse sulla monotonia non direi niente, perchè i limiti per alpha nullo e alpha grande non sono simmetrici... Dove abbiamo perso la simmetria tra le due popolazioni??
%
%Il caso \alpha = 0 vuol dire \beta_1^2=0, quindi il secondo membro di (73) è identicamente nullo. Abbiamo usato spesso 1/\alpha che non è accettabile per \alpha=0 o \alpha=\infty, quindi confermo che non citerei i due casi limite...

%If we want to quantify the approximation that we have done, we can prove the following statement
%\begin{theorem}
%\begin{equation}
%W_1(f_1,f_1^{(0)})=\mathcal{O}(\epsilon^2)
%\end{equation}
%\begin{proof}
%$W_1(f_1,f_1^{(0)})\le \int_{\R_+^2}|v-w|f_1(v)f_1^{(0)}(w) \, dv \,dw=\int_{\R_+^2}|v-w|f_1(v)(f_1(w)-\epsilon f_1^{(0)}(w) +\mathcal{O}(\epsilon^2) \, dv \,dw=\int_{\R_+^2}|v-w|f_1(v)f_1(w) \, dv \, dw -\epsilon \int_{\R_+^2}|v-w|f_1(v)f_1^1(w) \, dv \, dw +\mathcal{O}(\epsilon^2)$.
%Now $\int_{\R_+^2}|v-w|f_1(v)f_1^1(w) \, dv \, dw \le \int_{\R_+^2}(|v|+|w|)f_1(v)f_1^1(w) \, dv \, dw=M_1 \rho_1^1+\rho_1+M_1^1=0+0=0$, thanks to \eqref{eq:Hilbert_2}.
%\end{proof}
%\end{theorem}
\section{Numerical tests}
In this section we present some numerical tests that illustrate the dynamics of the model that we have introduced in the previous section. We integrate the kinetic equation \eqref{eq:boltz.fi_23} numerically using a modified version of
the Nanbu-Babovski Monte Carlo algorithm (see Algorithm \ref{alg:nanbu} in the Appendix). The latter is based on a direct implementation of the time discrete stochastic microscopic process \eqref{eq:micro.rules.gen}-\eqref{eq:bernoulli}-\eqref{def:g_gt} with \eqref{eq:T_lab.bin} as illustrated in Section \ref{subsec:int_and_switch} for $N$ agents, which in the limit $\Delta t \rightarrow 0^+$ produces the kinetic equation \eqref{eq:boltz.fi_23}. In particular, we shall consider the microscopic rules \eqref{eq:micro.rule_2}-\eqref{eq:P_mod}, with $n=2$. We perform an empirical statistics of the $N$ simulated agents and define the distribution functions $f_1^{MC},f_2^{MC}$, their masses $\rho_1^{MC},\rho_2^{MC}$ and first moments $M_1^{MC}, M_2^{MC}$.

\subsection{Interacting particles with label switch and exchange of physical quantities}
%We now show the numerical solutions of the model for international trade illustrated in Sec. \ref{sec:model}, i.e. the numerical simulations of the microscopic model \eqref{eq:micro.rules.gen}-\eqref{eq:bernoulli}-\eqref{def:g_gt}-\eqref{eq:T_lab.bin} as illustrated in Sec. \ref{subsec:int_and_switch} with the microscopic rules \eqref{eq:micro.rule_2}-\eqref{eq:P_mod}.

In all numerical tests we consider $\rho_1(0)=0.9$ and $\rho_2(0)=0.1$, i.e. the subgroup labelled with $x=1$ is initially more populated.
We remark that the assumed symmetry in the process \eqref{eq:P_mod} implies $P_{11}^{11}=P_{11}^{22} = P_{22}^{22} = P_{22}^{11}=0$ and then $P_{11}^{12}=P_{11}^{21}=0.5$, being $P_{ij}^{kl}$ a conditional probability. In the following numerical tests, we shall always consider $\omega_1=\omega_2=0.5$, as the values of $\omega_1$ and $\omega_2$ do not affect the averages' evolution and stationary state \eqref{medie}.
Moreover, we fix
\[
\lambda_{11}=1
\]
and we vary the other parameters.

In the first set of simulations ({\bf Test 1} in the following), we choose $\lambda_{22}=10, \lambda_{12}=1$, $P_{12}^{11}=0.5, P_{12}^{22}=0.5$ and we consider two different initial conditions for the distributions $f_1$ and $f_2$. In {\bf Case A} we have that the first population is poorer than the second population at time $t=0$, i.e.
\begin{equation}\label{def:f12_0}
f_1(v,0)=\rho_1(0)\mathbb{1}_{[0,1]}(v), \qquad f_2(v,0)=\rho_2(0)\dfrac{1}{10}\mathbb{1}_{[5,15]}(v)
\end{equation}
while in {\bf Case B} we invert the initial wealths
\[
f_1(v,0)=\rho_1(0)\dfrac{1}{10}\mathbb{1}_{[5,15]}(v), \qquad f_2(v,0)=\rho_2(0)\mathbb{1}_{[0,1]}(v),
\]
i.e. the second population is poorer than the first population at time $t=0$.
We report the results in Figure \ref{fig1}.
First of all, we observe that this choice of parameters prescribes $\alpha<1$, that, as showed by the macroscopic equations \eqref{eq:masse}, implies $\rho_1^{\infty}>\rho_2^{\infty}$ (see Figure \ref{fig1}(a)). As forecast by theoretical results, the final average wealth $m^{\infty}=m_1=m_2$ is closer to the initial wealth of the initially more populated subgroup: then $m^{\infty}$ is smaller in case A) and larger in case B) (see Figure \ref{fig1} (b)). This implies a different behaviour of the first moment of both $f_1$ and $f_2$ (see Figure \ref{fig1} (c)): while in case B) the first population remains the richer one as it is the one that is initially more populated, in scenario A), the mean wealth is inverted as the first population becomes reacher. In each case we compare the evolution of the macroscopic quantities $\rho_1^{MC},\rho_2^{MC},M_1^{MC},M_2^{MC},m_1^{MC},m_2^{MC}$ as prescribed by the microscopic model \eqref{eq:micro.rules.gen}-\eqref{eq:bernoulli}-\eqref{def:g_gt}-\eqref{eq:T_lab.bin} with \eqref{eq:micro.rule_2}-\eqref{eq:P_mod} and the ones whose evolution is given by the derived equations \eqref{eq:masse}-\eqref{eq:momenti}-\eqref{medie} for the macroscopic quantities $\rho_1,\rho_2,M_1,M_2,m_1,m_2$. Being $\Delta t=1e-2$ and $N=10^6$, we observe a very good agreement between the solution of the microscopic model and the one of the macroscopic model.
The integration of the kinetic equations also allows to approximate numerically the distribution functions $f_1$ and $f_2$ that we report at the equilibrium in Figure
\ref{fig1}(d) in both cases A) and B). In Figure \ref{fig1} (e)-(f) we report the time evolution of $f_1^{MC}$ and $f_2^{MC}$ in case A) with initial condition \eqref{def:f12_0}.  In Fig. \ref{fig1} (e)-(f) we report the time evolution for $f_1^{MC}, f_2^{MC}$.

%\subsubsection{$v$-dependent $P_{ij}^{kl}$}
%We now consider the case in which the transfer probability $P_{ij}^{kl}$ may depend on the microscopic physical quantity $v,w$ of the interacting agents. This allows to link the probability of a transfer to the actual wealth of the interacting agents. In particular, we shall consider two cases:
%\begin{itemize}
%\item[$a)$] $P_{ij}^{kl}(v,w)=\tilde{P}_{ij}^{kl}\dfrac{1}{2}\left[\exp^{-v}+\exp^{-w}\right]$: a case in which both agents are more likely to migrate if they have a small wealth;
%\item[$b)$] $P_{ij}^{kl}(v,w)=\tilde{P}_{ij}^{kl}\dfrac{1}{2}\left[\dfrac{1}{2}\left(1-\exp^{-v}+1-\exp^{-w}\right)\right]$: a case in which both agents are more likely to migrate if they have a large wealth.
%\end{itemize}

\begin{figure}[!htbp]
\centering
\subfigure[]{\includegraphics[width=0.48\textwidth]{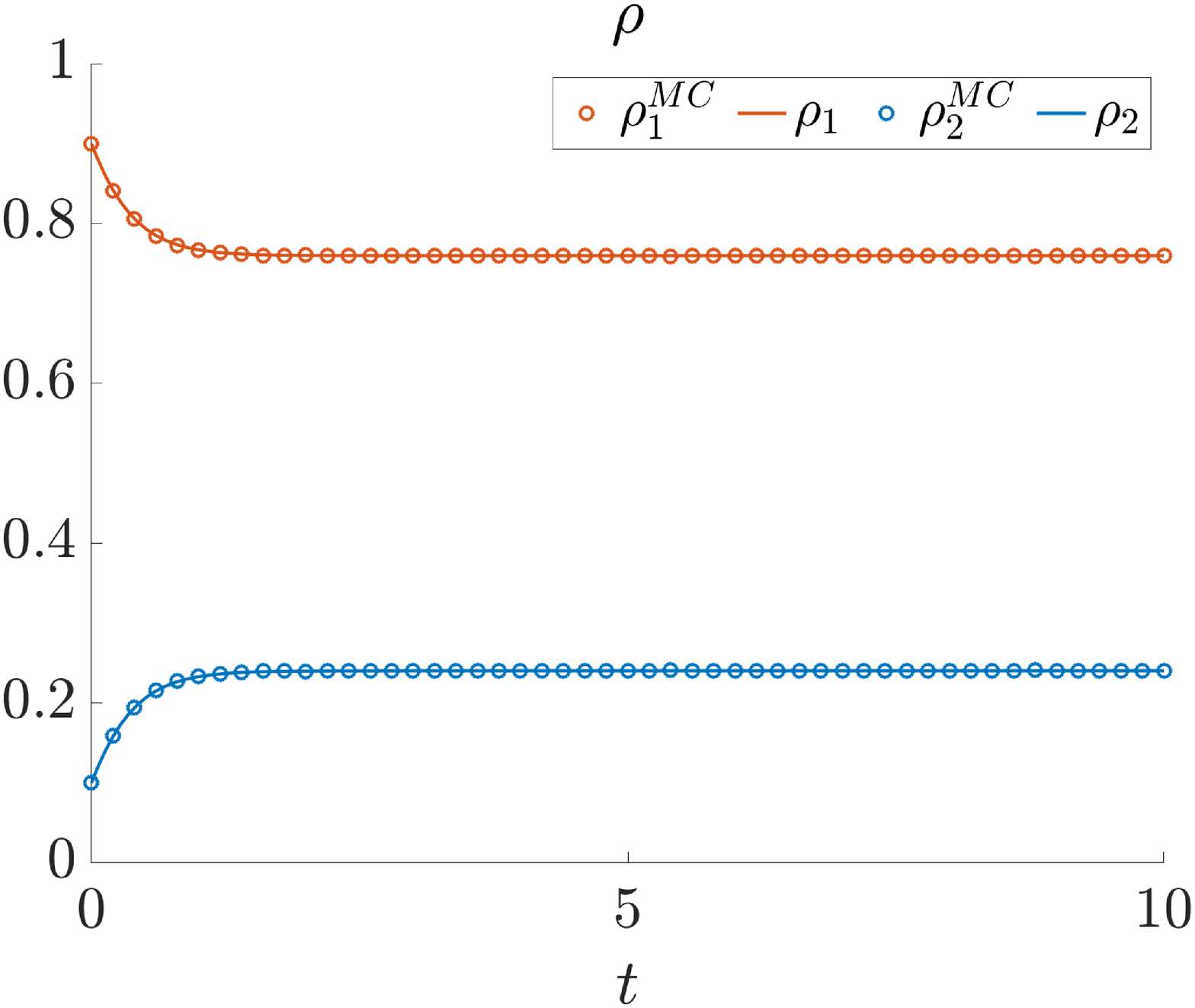}}
\subfigure[]{\includegraphics[width=0.48\textwidth]{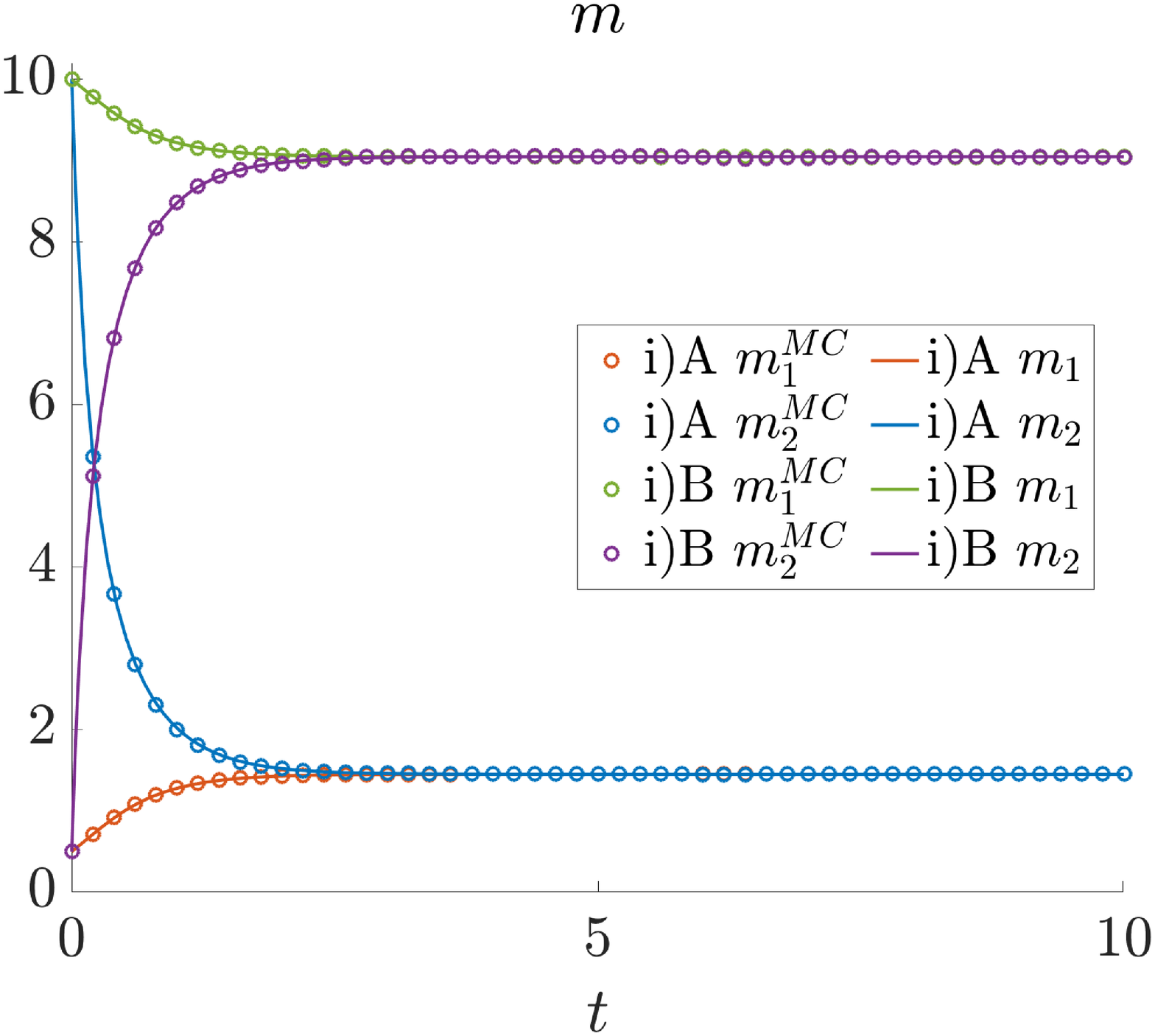}}\\
\subfigure[]{\includegraphics[width=0.48\textwidth]{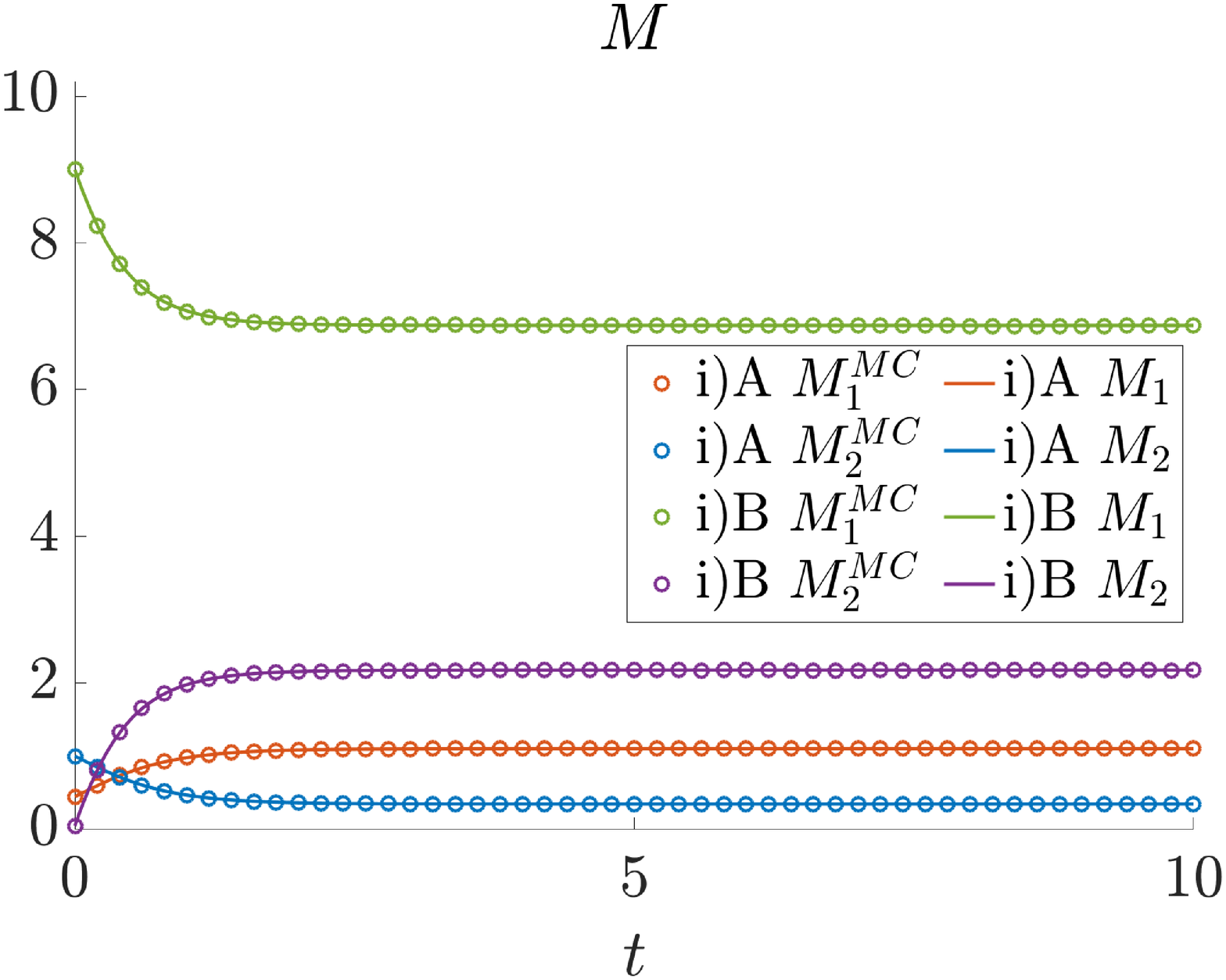}}
\subfigure[]{\includegraphics[width=0.48\textwidth]{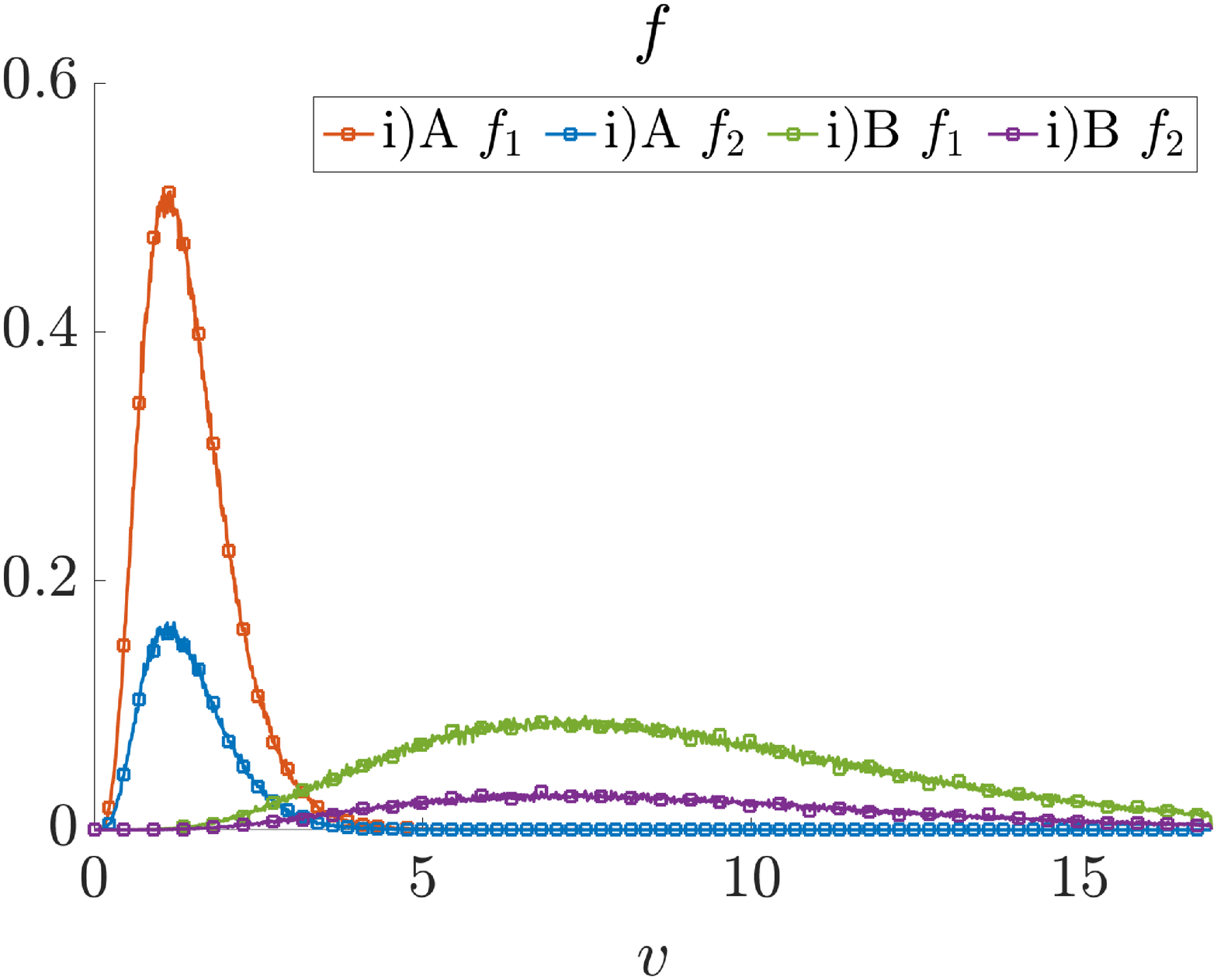}}\\
\subfigure[]{\includegraphics[width=0.48\textwidth]{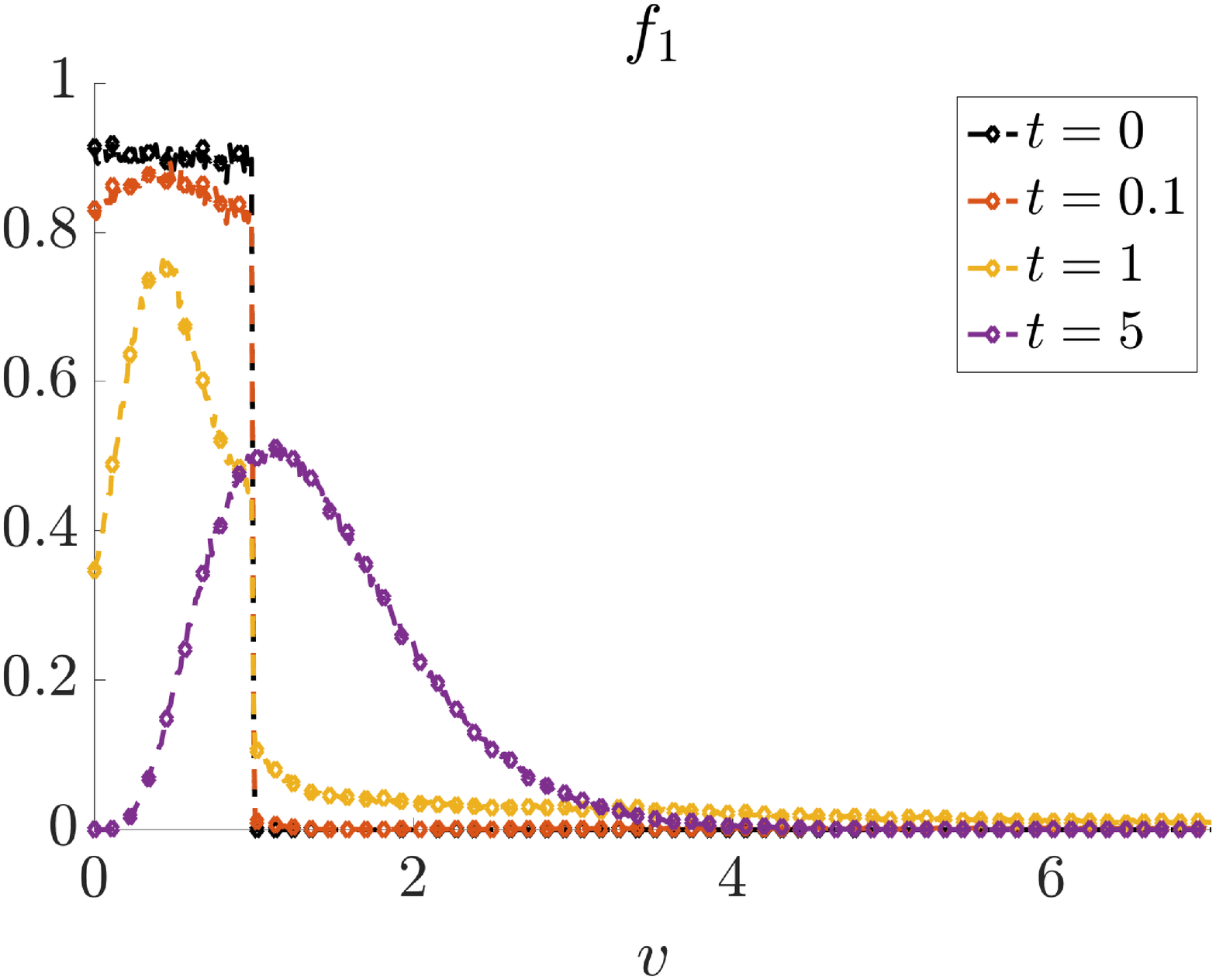}}
\subfigure[]{\includegraphics[width=0.48\textwidth]{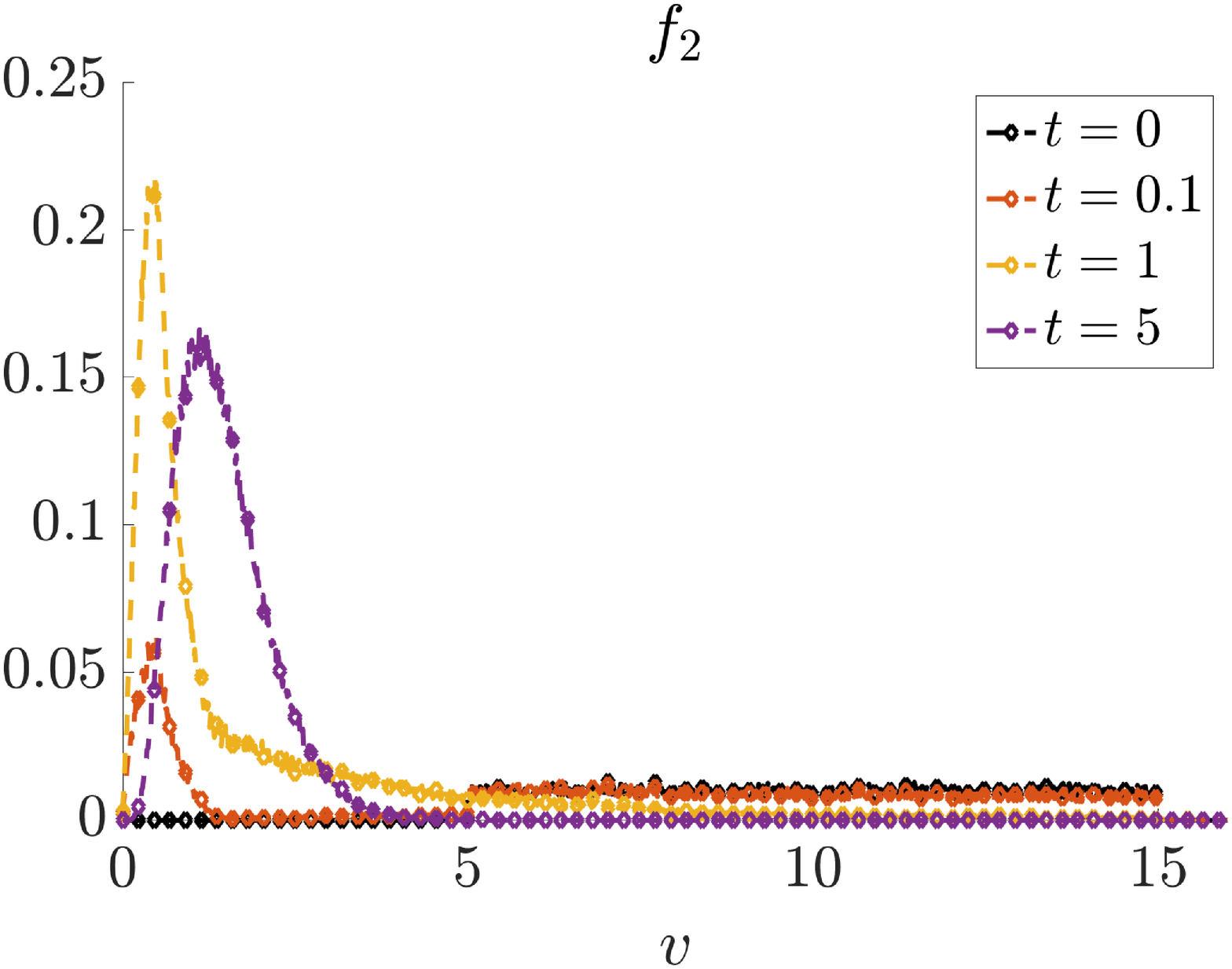}}
\caption{{\bf Test 1}: $\lambda_{11}=1, \lambda_{22}=10, \lambda_{12}=1$, $P_{12}^{11}=0.5, P_{12}^{22}=0.5$. In all figures we report the time evolution of the macroscopic quantities as prescribed by the kinetic model and by the derived macroscopic equations: masses (a), averages (b), first moments (c). Straight lines correspond to the solution to the macroscopic equations \eqref{eq:masse}-\eqref{medie}-\eqref{eq:momenti}, while circles correspond to the solution of the kinetic equation that we obtain by simulating with the Monte Carlo algorithm \ref{alg:nanbu} the microscopic model \eqref{eq:micro.rules.gen}-\eqref{eq:bernoulli}-\eqref{def:g_gt}-\eqref{eq:T_lab.bin} as illustrated in Sec. \ref{subsec:int_and_switch} with the microscopic rules \eqref{eq:micro.rule_2}-\eqref{eq:P_mod}. In all figures we compare the results of the solutions given the two different initial conditions A and B. In figure (d) we show the steady states of distributions $f_1$, $f_2$ in both test cases A) and B), while in figures (e), (f) we report time evolution of distributions functions in the test case A).
}
\label{fig1}
\end{figure}
In the second set of simulations ({\bf Test 2}) we choose $\lambda_{11}=1, \lambda_{22}=1, \lambda_{12}=10$, i.e. intra-species interactions have the same frequency in the two populations, while the inter-group interactions have a higher frequency. Moreover, we consider three cases for the inter-group interactions: in case i) $P_{12}^{11}=0.5, P_{12}^{22}=0.5$, i.e. given an inter-group interaction, the probability for both agents of transferring is the same, while in case ii) $P_{12}^{11}=0.2, P_{12}^{22}=0.8$, i.e. the probability of transferring to the subgroup 2 is higher and iii) $P_{12}^{11}=0.8, P_{12}^{22}=0.2$, i.e. the probability of transferring to the subgroup 1 is higher. The initial condition is set as in \eqref{def:f12_0}. We observe that in the three cases i), ii) and iii) we have respectively $\alpha=1, \alpha >1, \alpha <1$, that imply, see \eqref{eq:masse}, $\rho_1^{\infty}=\rho_2^{\infty}, \rho_1^{\infty}<\rho_2^{\infty}, \rho_2^{\infty}<\rho_1^{\infty}$, respectively. In particular, because of \eqref{eq:mass_switch}, in case ii) we have a switch in the trend of the populations, as population 1 becomes the less populated (see Fig. \ref{fig2} (a)). This also implies a different (non monotone) trend of the first moments as reported in Fig. \ref{fig2} (b).
\begin{figure}[!htbp]
\subfigure[]{\includegraphics[width=0.5\textwidth]{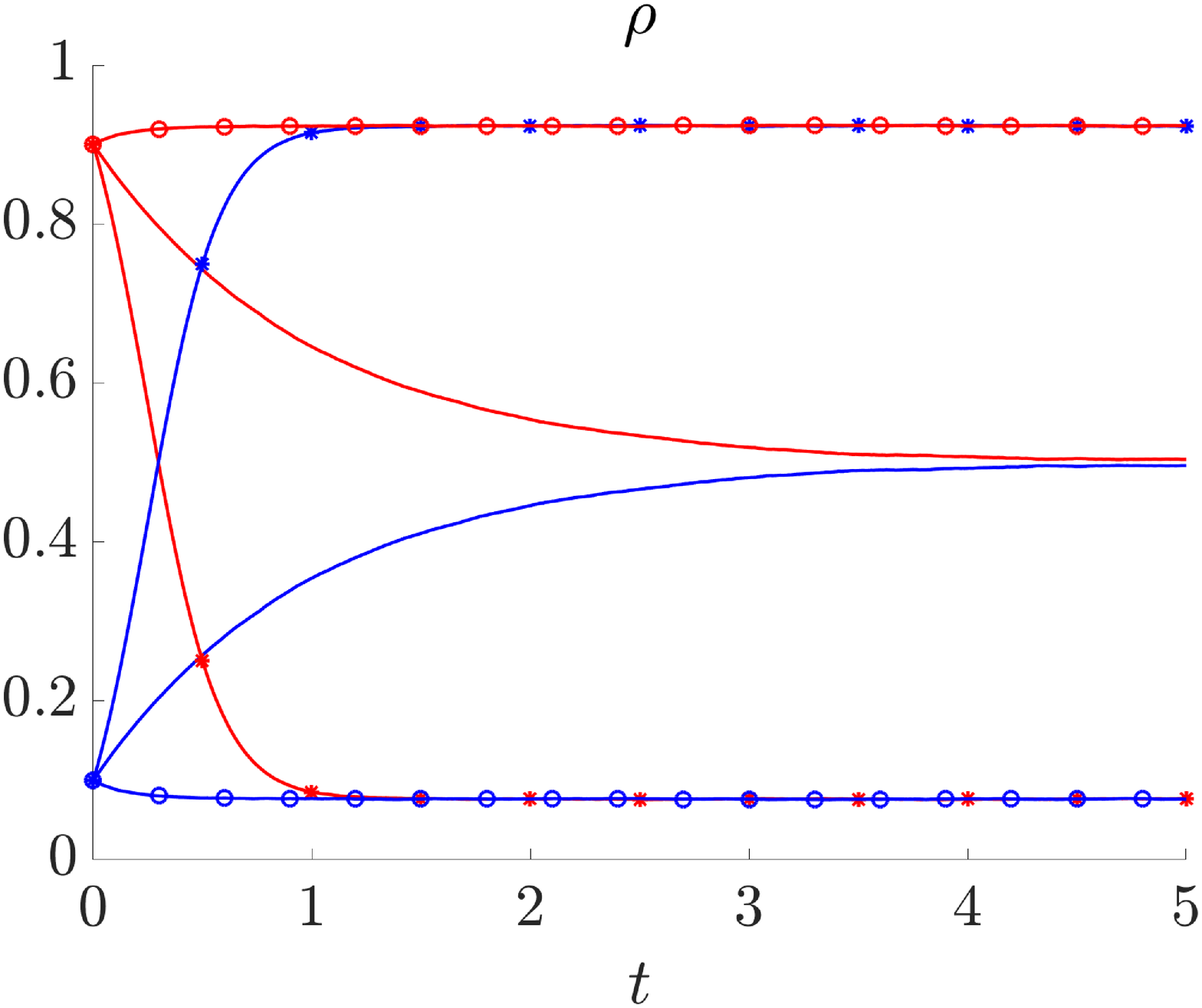}}
\subfigure[]{\includegraphics[width=0.5\textwidth]{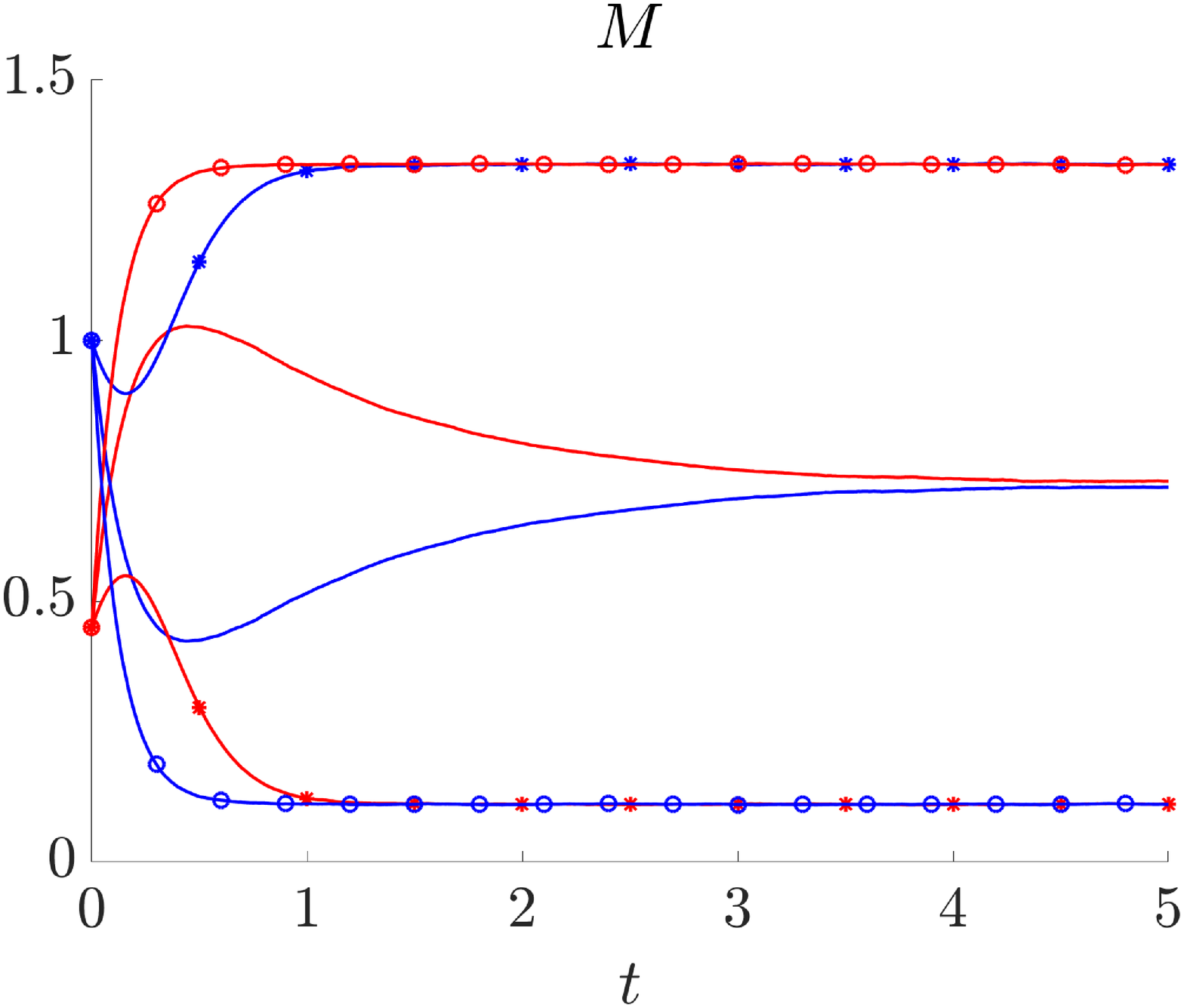}}
\caption{{\bf Test 2.} Solutions to the macroscopic model \eqref{eq:masse}-\eqref{eq:momenti}. In (a) we report the masses $\rho_1, \rho_2$ and in (b) the first moments $M_1, M_2$. Population 1 is in \textcolor{red}{red},while population 2 is in \textcolor{blue}{blue}. The parameters are $\lambda_{11}=1, \lambda_{22}=1, \lambda_{12}=10$ and we show three cases: i) (straight lines) $P_{12}^{11}=0.5, P_{12}^{22}=0.5$,  ii) (* marker) $P_{12}^{11}=0.2, P_{12}^{22}=0.8$  and iii) (o marker) $P_{12}^{11}=0.8, P_{12}^{22}=0.2$. }\label{fig2}
\end{figure}

In {\bf Test 3}, we consider a switching probability depending on the microscopic wealth. In this case, it is not immediate, in general, to derive equations for the macroscopic quantities, unless in special cases, for example $P_{ij}^{kl}$ having a linear dependence on $v$, or by imposing a monokinetic closure \cite{pareschi2013BOOK}. In Figure \ref{fig3}  we have that  $\lambda_{11}=.1, \lambda_{22}=1, \lambda_{12}=10$ and $P_{12}^{11}=0.2\dfrac{1}{4}\left[1-\exp^{-v}+1-\exp^{-w}\right], P_{12}^{22}=0.8\dfrac{1}{4}\left[1-\exp^{-v}+1-\exp^{-w}\right]$. We also present a comparison with the solution of  macroscopic equations \eqref{eq:masse}-\eqref{medie}-\eqref{eq:momenti} (straight lines) where we consider constant switching probabilities $P_{12}^{11}=0.2, P_{12}^{22}=0.8.$ We observe that both the microscopic model with $v$-dependent switching probabilities and the macroscopic model with constant switching probabilities forecast a similar behavior of the macroscopic quantities in the long run. The microscopic model with a $v$-dependent switching probability forecasts the same behaviour but with a delay, and this is due to the fact that the switching probabilities are smaller than the constant ones. In magenta and green we also present the results of the simulation of the  microscopic model in case $P_{12}^{11}=0.2\dfrac{1}{2}\left[\exp^{-v}+\exp^{-w}\right], P_{12}^{22}=0.8\dfrac{1}{2}\left[\exp^{-v}+\exp^{-w}\right]$. In this case we can observe that the convergence is even slower. This is due to the fact that for large values of the wealth $v$, the switching probability is very small.

\begin{figure}[!htbp]
\includegraphics[width=0.48\textwidth]{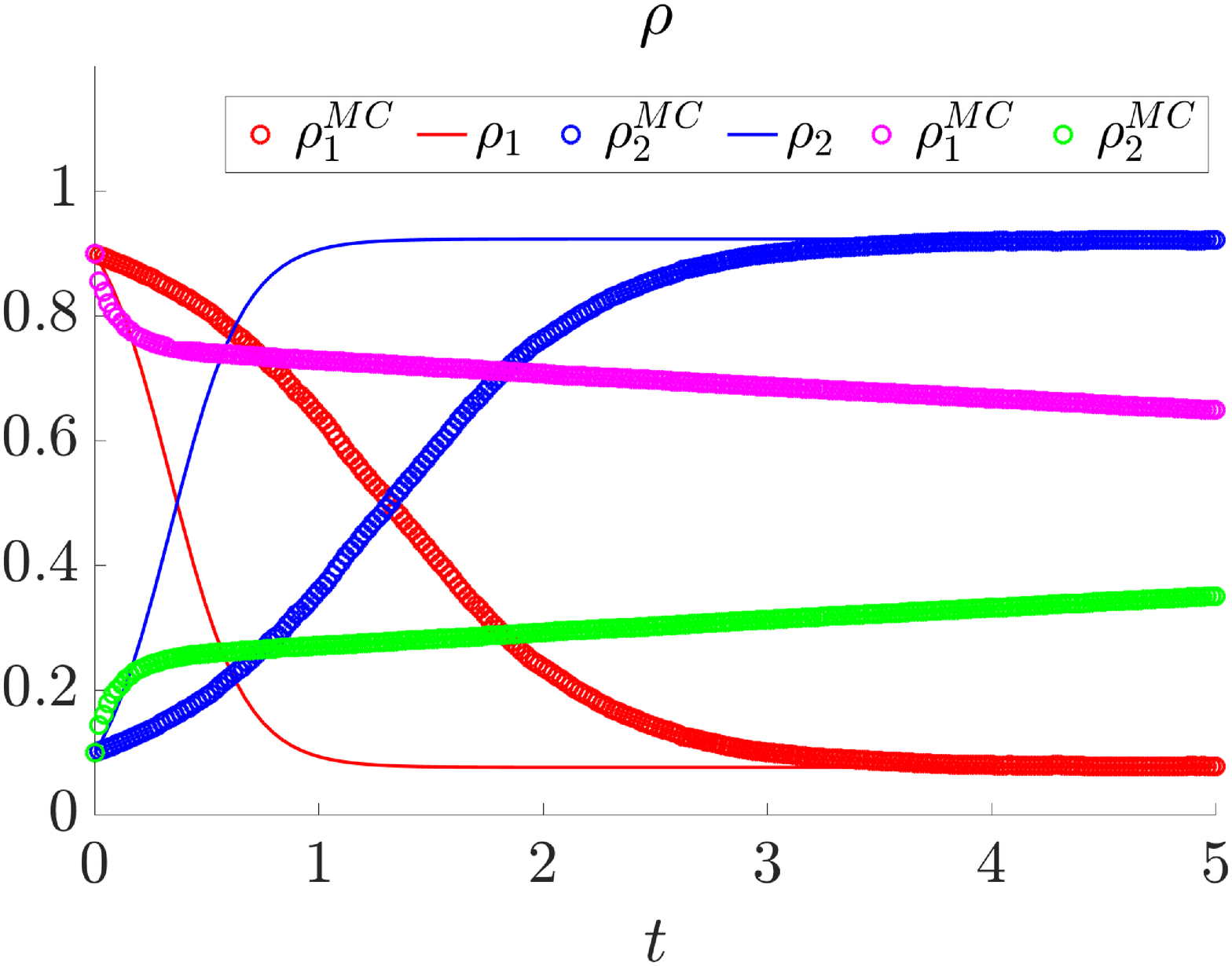}
\includegraphics[width=0.48\textwidth]{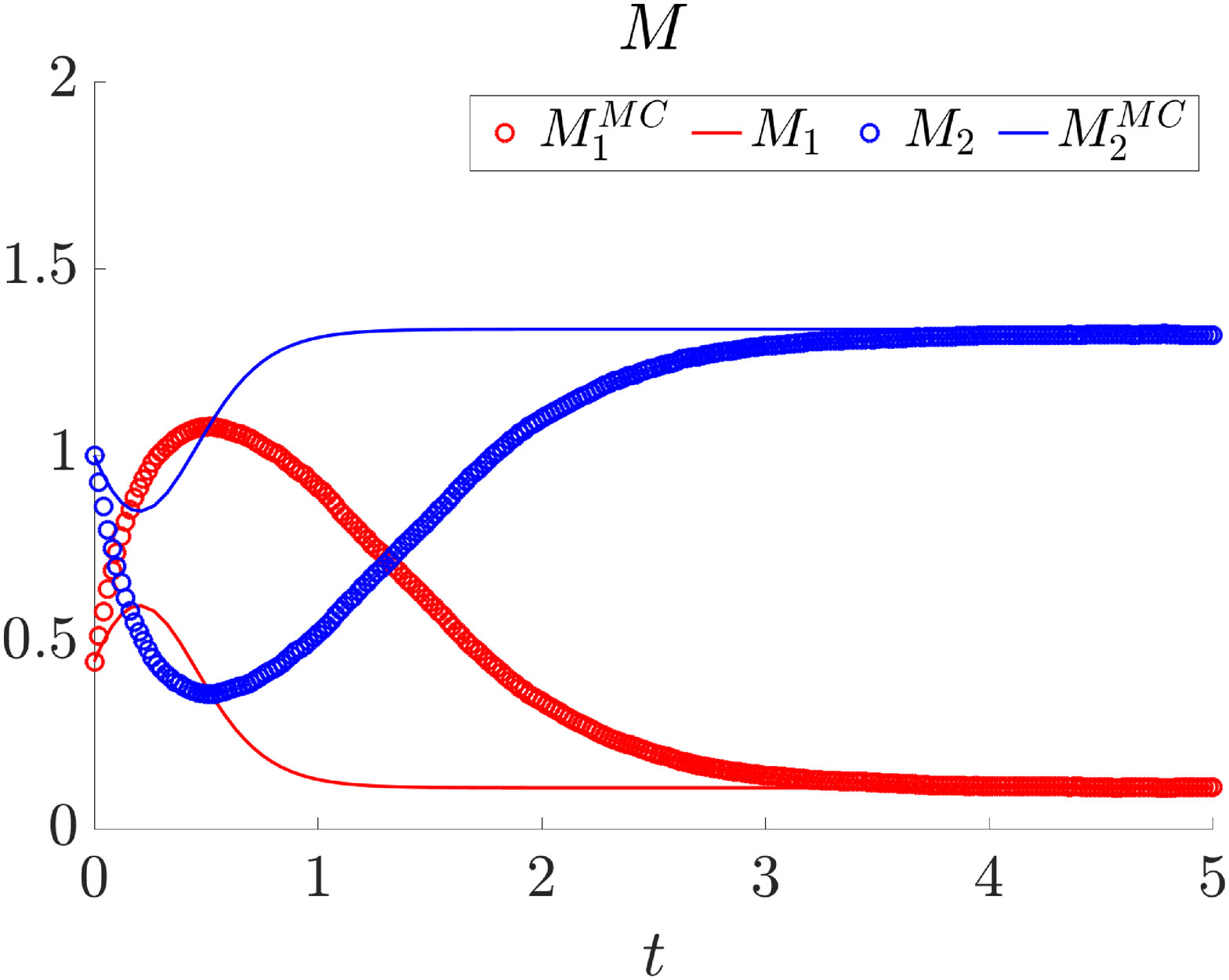}\\
\includegraphics[width=0.48\textwidth]{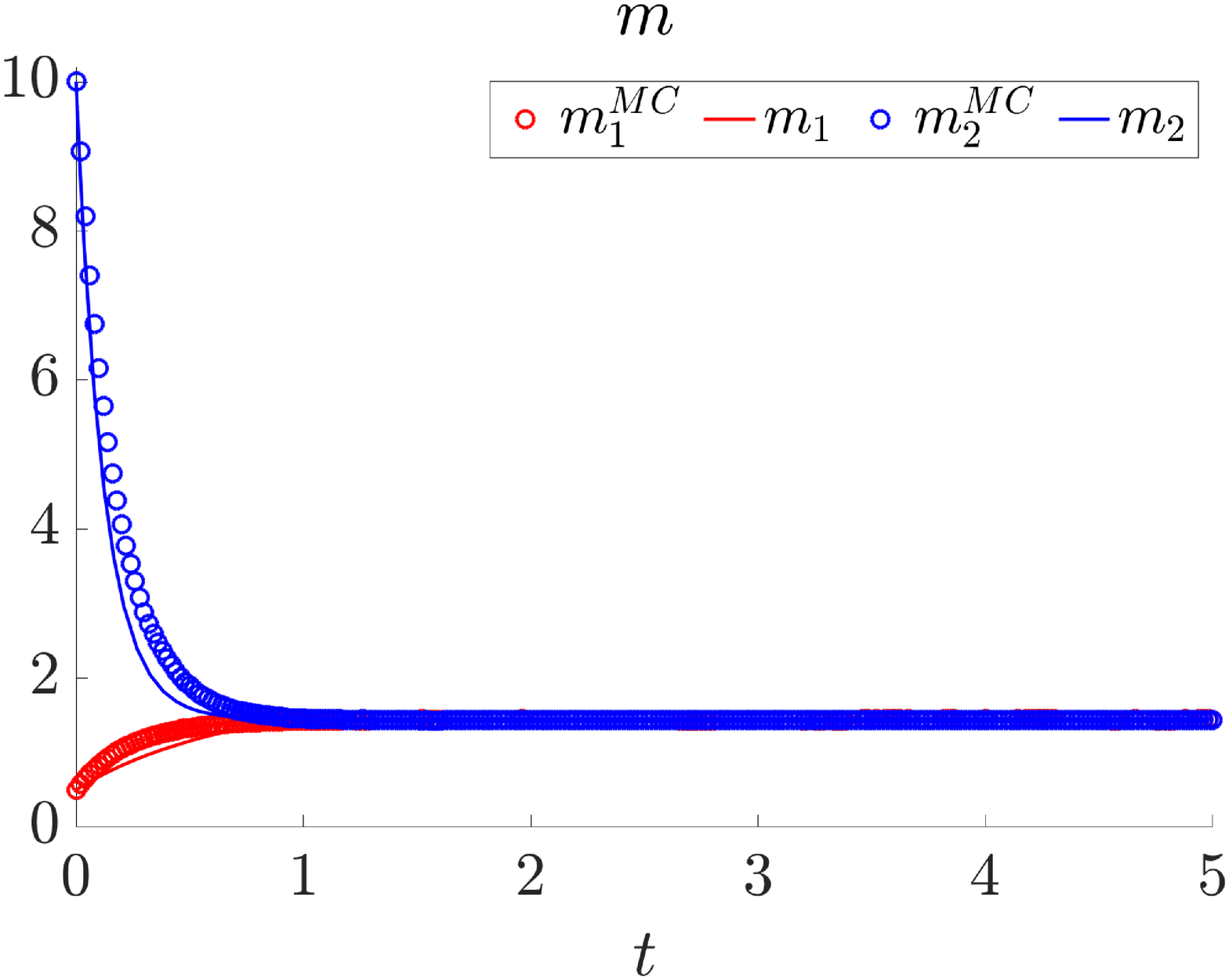}
\includegraphics[width=0.48\textwidth]{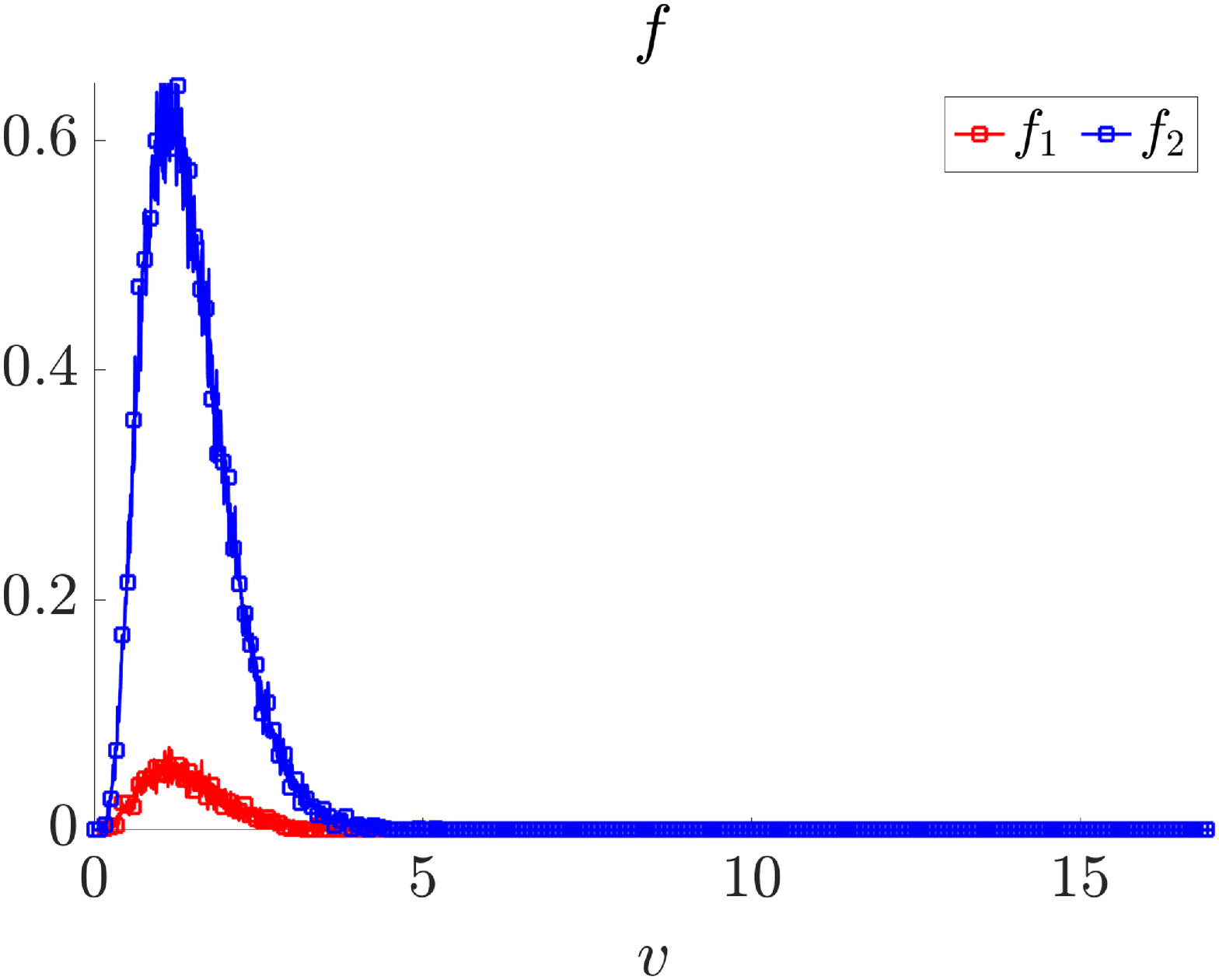}
\caption{{\bf Test 3.} Solution to the microscopic model \eqref{eq:micro.rules.gen}-\eqref{eq:bernoulli}-\eqref{def:g_gt}-\eqref{eq:T_lab.bin} as illustrated in Sec. \ref{subsec:int_and_switch} with the microscopic rules \eqref{eq:micro.rule_2}-\eqref{eq:P_mod} with the Monte Carlo algorithm \ref{alg:nanbu} (circles). Here $\lambda_{11}=1, \lambda_{22}=.1, \lambda_{12}=10$ and $P_{12}^{11}=0.2\dfrac{1}{2}\left[1-\exp^{-v}+1-\exp^{-w}\right], P_{12}^{22}=0.8\dfrac{1}{2}\left[1-\exp^{-v}+1-\exp^{-w}\right]$. We also present a comparison with the solution of  macroscopic equations \eqref{eq:masse}-\eqref{medie}-\eqref{eq:momenti} (straight lines) where we consider constant switching probabilities $P_{12}^{11}=0.2, P_{12}^{22}=0.8.$ In magenta and green we also present the results of the simulation of the  microscopic model in case $P_{12}^{11}=0.2\dfrac{1}{2}\left[\exp^{-v}+\exp^{-w}\right], P_{12}^{22}=0.8\dfrac{1}{2}\left[\exp^{-v}+\exp^{-w}\right]$.}
\label{fig3}
\end{figure}

\subsection{Quasi-invariant regime and Fokker-Planck equation}
In this section, we consider the quasi-invariant regime \eqref{eq:trans_qinv.rules}, i.e. we analyse the dynamics on a long time-scale by considering small exchanges of wealth, while the switching probability is not rescaled. In this framework, we have seen that it is possible to approximate the leading order of the stationary solution $f_1^{(0)}$ through \eqref{stato_staz}. Here, we compare the stationary state \eqref{stato_staz} with the solution $f_1^{MC}$  obtained by the numerical integration of the microscopic process \eqref{eq:micro.rules.gen}-\eqref{eq:bernoulli}-\eqref{def:g_gt}-\eqref{eq:trans_qinv.rules_noncons} with the microscopic rules \eqref{eq:micro.rule_2}-\eqref{eq:P_mod}, where in the quasi-invariant regime \eqref{eq:trans_qinv.rules_noncons} we have chosen $\epsilon=10^{-3}$. In Figure \ref{fig4} we represent the analytical $\tilde{f}_1^{(0)}$ as given in \eqref{stato_staz} and the approximation of $f_1^{MC}$. We can remark that, despite the fact that $f_1^{MC}$ is obtained through a Monte Carlo simulation and $\tilde{f}_1^{(0)}$ is an approximation, the agreement is quite good. In the right panel, we also represent the numerical approximation of the distribution function $g_1^{MC}$ of the first population, in case we consider a quasi-invariant transition probability defined by \eqref{eq:trans_qinv.rules}. With this choice it is granted that the evolution of both the average and the energy in the quasi-invariant regime $\epsilon \ll 1$ is the same as in the standard regime defined by $\epsilon=1$. In this case it was not possible to determine the stationary state explicitly like in the quasi-invariant regime leading to \eqref{stato_staz} that, on the other hand, only grants that the average is the same for each $\epsilon$.
\begin{figure}
\centering
\includegraphics[width=0.48\textwidth]{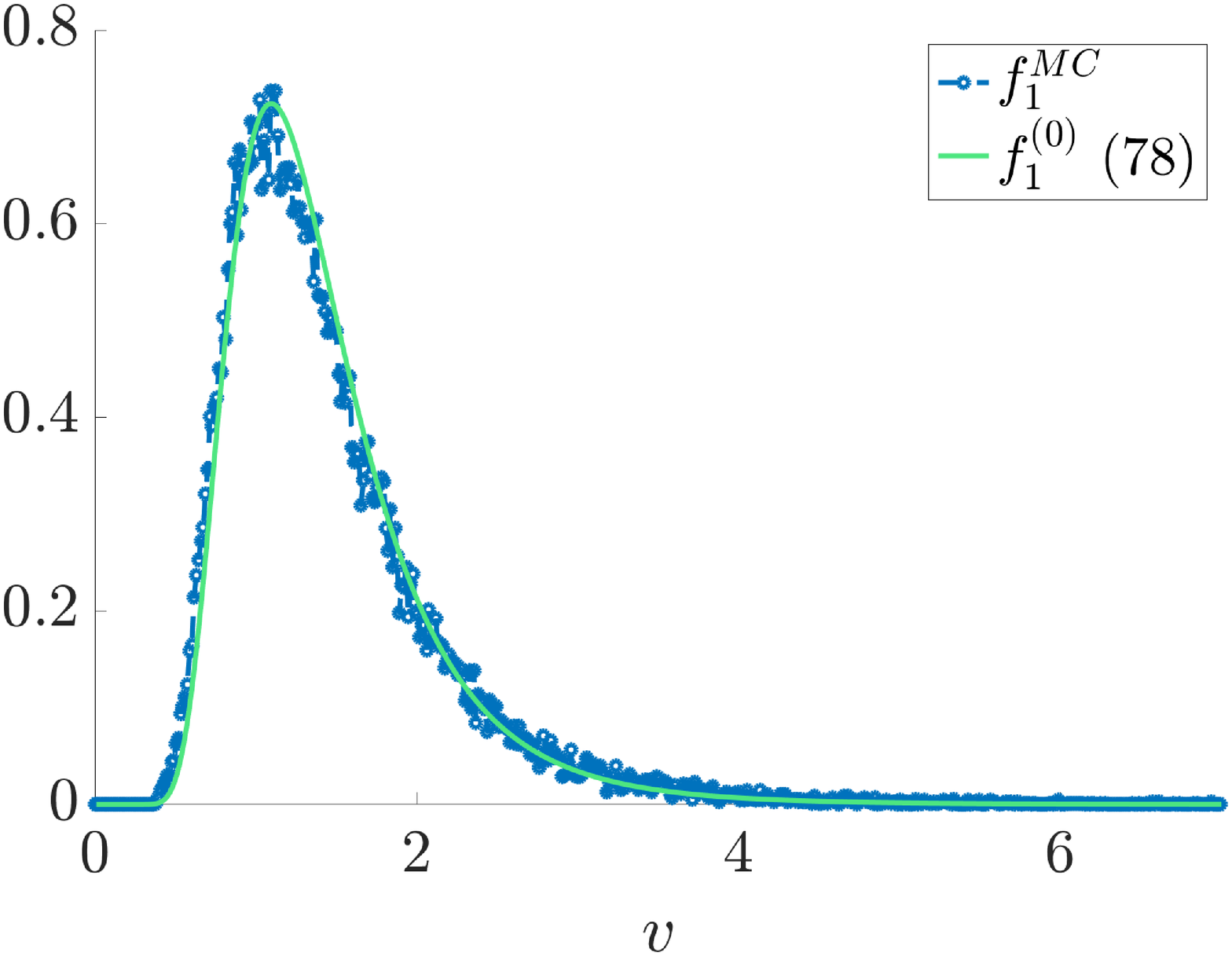}
\includegraphics[width=0.48\textwidth]{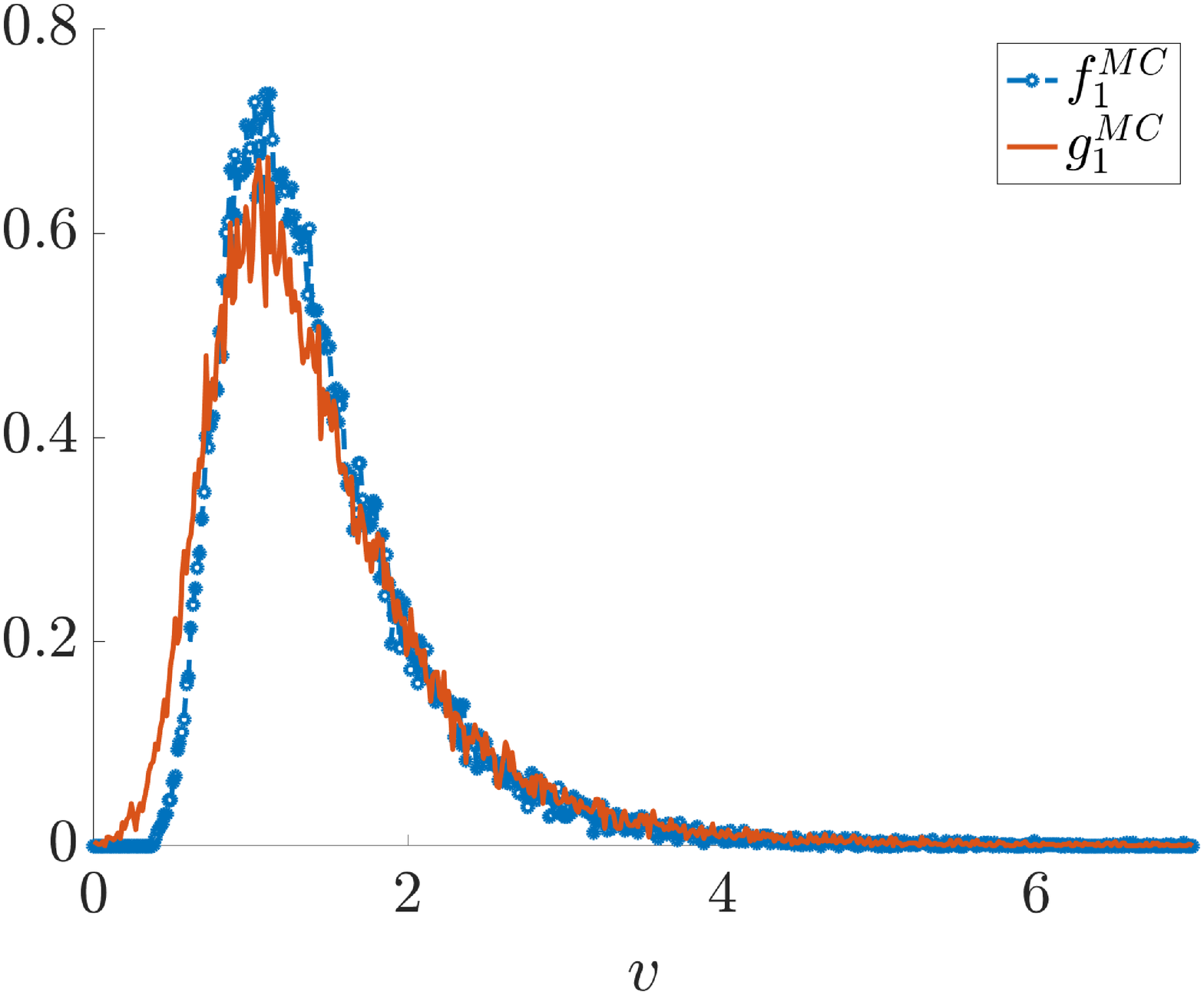}
\caption{Solution of the Fokker-Planck equation \eqref{eq:FP_simpl}. Left: comparison between the approximated stationary state $\tilde{f}_1^{(0)}$ given by \eqref{stato_staz} and the solution $f_1^{MC}$ of the microscopic model \eqref{eq:micro.rules.gen}-\eqref{eq:bernoulli}-\eqref{def:g_gt}-\eqref{eq:trans_qinv.rules_noncons} with the microscopic rules \eqref{eq:micro.rule_2}-\eqref{eq:P_mod}, with $\epsilon=10^{-3}$ in the quasi-invariant regime \eqref{eq:trans_qinv.rules_noncons} obtained with $N=10^6, \Delta t=10^{-3}$. Right: comparison between $f_1^{MC}$ and $g_1^{MC}$ obtained integrating the same, with the quasi-invariant regime \eqref{eq:trans_qinv.rules}.}
\label{fig4}
\end{figure}

\section{Conclusions}
In this paper we have presented a general framework for modeling systems of interacting particles with multiple microscopic states changing simultaneously according to a given dynamics.
In particular, the microscopic description relies on Markovian processes described by transition probabilities, as they depend on the pre-interaction states, and the interaction frequency depends on the microscopic states. The fact of starting from the microscopic stochastic process allows to describe in more detail the dynamics, by including parameters and quantities related to the phenomenon under study that can be observed. The derivation, through kinetic equations, of macroscopic equations allows to obtain also at the aggregate level a higher level of detail that is inherited from the underlying microscopic dynamics.

\noindent Under some assumptions, general results concerning well-posedness, existence and uniqueness of a solution for the Cauchy problem associated to our kinetic equation have been shown. We have also rephrased the concept of quasi-invariant limit in the present framework, leading to evolution equations of Fokker-Planck type.

We have applied the present modeling framework in order to describe systems of binarily interacting agents characterized by a physical quantity $v$ (representing wealth, or opinion, or viral load, etc.)\ and by a label $x$ denoting the belonging to a given subgroup. The physical quantity changes according to binary interaction rules, while the label changes, simultaneously with the physical quantity, through a switch process caused by the same binary interaction.
In this context, we have seen that the description of the microscopic process by means of transition probabilities allows us to remove the reversibility assumption on the interaction rule, modeling thus also stochasticity in the binary encounters giving rise to transfers (not present in the paper \cite{bisi2021PhTB} using classical Boltzmann operators analogous to the reactive ones). Moreover, in our framework it is easier to consider a non-constant switching probability, i.e. depending on the microscopic physical quantity.
We have analyzed and discussed various quasi-invariant regimes and performed some numerical tests showing a very good agreement between the microscopic Monte Carlo simulations and the derived macroscopic equations.

The modelling framework investigated in this paper is worth to be applied and generalized to many other problems. As first, the model for international trade with transfers presented in Section 4 could be extended by adding an extra independent microscopic process for $v$, describing the exchange of goods without transfers; this would make the model even more similar to the kinetic description of gaseous mixtures, where elastic collisions (which do not change the nature of the particles) coexist with chemical reactions (changing the species of the reacting particles).
Epidemic models based on a kinetic approach could be improved owing to our stochastic framework with multiple states as it allows to start from a microscopic description and to consider independent or simultaneous microscopic stochastic dynamics for the different variables of the microscopic state. For example, the so-called ``non--conservative'' interactions giving rise to the passage from one compartment to another could be made more realistic taking into account also the simultaneous change of viral load (of individuals) as done in \cite{dellamarca2022NHM,dellamarca2022preprint} or internal activity (of cells).
Moreover, the relation with kinetic models with label switching and gradient descent could be established \cite{burgerrossi}.
Eventually, applications to situations with many internal states is the final scope of our framework. It could provide for instance a physically reasonable description of mixtures of polyatomic gases, with each molecule characterized by its species label $i$, its velocity ${\bf v} \in \mathbb{R}^3$, and its internal energy that could also be separated into the vibrational part (typically described by a discrete variable) and the rotational part (typically approximated by a continuous variable) \cite{herzberg1950BOOK, borsoni2022CIMP}. Even in econophysics, the possible influence of the personal knowledge of the market on the strategy adopted in the trades (as sketched in \cite{pareschi2014PhTB} for a single population) could be described considering the individual knowledge as an additional microscopic state, besides the population label and the individual amount of wealth.
Suitable quasi-invariant limits and properties of steady states of such non-standard kinetic descriptions of various interacting populations are completely open problems worth to be investigated in future research.

\textbf{Acknowledgments}
This research was initiated during the post-doc contract of N.L. at the Department of Mathematical, Physical and Computer Sciences of Parma University, funded by the Italian National Research Project ``Multiscale phenomena in Continuum Mechanics: singular limits, off-equilibrium and transitions'' (Prin 2017YBKNCE). The authors also thank the support by University of Parma, by Politecnico di Torino, by the Italian National Group of Mathematical Physics (GNFM-INdAM), and by the Italian PRIN Research Project ``Integrated Mathematical Approaches to Socio--Epidemiological Dynamics'' (Prin 2020JLWP23, CUP: E15F21005420006).

\section*{Appendix: Nanbu-Babovski algorithm}

\begin{algorithm}[H]
	\caption{Nanbu-Babovski algorithm with mass transfer for model~\eqref{eq:micro.rules.gen}-\eqref{eq:bernoulli}-\eqref{def:g_gt}- \eqref{eq:micro.rule_2}-\eqref{eq:P_mod}}
	\label{alg:nanbu}
	\KwData{
		\begin{itemize}[noitemsep]
			\item $N\in\mathbb{N}$ total number of agents of the system;
			\item $N_1^n,\,N_2^n\in\mathbb{N}$ numbers of agents in $x=1$, $x=2$, respectively, at time $t^n:=n\Delta{t}$ and $v_1^n, \, v_2^n$ the microscopic states of agents in $x=1$, $x=2$, respectively, at time $t^n$;
			%\item $\rho_1^n=N_1^n/N \, \rho_2^n=N_1^n/N$;
			%\item compute $m_1^n$, $m_1^n$ as the average of $v_1^n, \, v_2^n$;
			%\item $M_1^n=m_1^n\rho_1^n$, $M_2^n=m_2^n\rho_2^n$.
		\end{itemize}
	}
	Fix $\Delta{t}\leq\min\{\frac{1}{\max \lambda_{ij}}\}$\;
	\For{$n=0,\,1,\,2,\,\dots$}{
		Compute \newline{}
		$\rho_1^{MC,n}=\dfrac{N_1^n}{N}, \qquad \rho_2^{MC,n}=\dfrac{N_2^n}{N}, \qquad
			m_1^{MC,n}=\dfrac{1}{N_1^n}\displaystyle{\sum_{k=1}^{N_1^n}}v_k^n, \qquad m_2^{MC,n}=\dfrac{1}{N_2^n}\displaystyle{\sum_{k=1}^{N_2^n}}v_k^n$\;
		\Repeat{no unused pairs of agents are left}{
			Pick randomly two agents $(x_i^n,\,v_i^n)$, $(x_j^n,\,v_j^n)$ with $i\neq j$\;
			\For{$h=i,\,j$}{Sample $\Theta\sim\operatorname{Bernoulli}(\lambda_{x_i^nx_j^n}\Delta{t})$\;
				\If{$\Theta=1$}{
				\For{$\lbrace x'_i,\,x'_j\rbrace \in \cI_n^2$}{
			    Sample $J\in\{1,\,0\}$ with law
				    $\P(J=1)=P_{x_i^n,x_j^n}^{x_i'x_j'}, \quad \P(J=0)=1-P_{x_i^n,x_j^n}^{x_i'x_j'}$;\newline{}
				    \If{$J=1$}{
                   Set $(x_i^{n+1},x_j^{n+1})=(x_i',x_j')$;\newline{}
                   Set $(v_i^{n+1},\,v_j^{n+1})=(v_i',v_j')$ where $(v_i',v_j')$ is given by \eqref{eq:bin_rules_gen}-\eqref{eq:micro.rule_2} and \textbf{break}}
                   }}				
					    \Else{Set $x_h^{n+1}=x_h^n, \, v_h^{n+1}=v_h^n$;}
					
			}
			}
			}
\end{algorithm}

\bibliographystyle{plain}
\bibliography{biblio}

\begin{thebibliography}{10}

\bibitem{ambrosio}
L.~Ambrosio, N.~Gigli, and G.~Savar\`e.
\newblock {\em Gradient flows in metric spaces and in the space of probability
  measures}.
\newblock Birkhauser Verlag, Basel, 2008.

\bibitem{andries2002JSP}
P.~Andries, K.~Aoki, and B.~Perthame.
\newblock A consistent bgk-type model for gas mixtures.
\newblock {\em J. Stat. Phys.}, 106:993--1018, 2002.

\bibitem{bisi2017UMI}
M.~Bisi.
\newblock Some kinetic models for a market economy.
\newblock {\em Boll. Unione Mat. Ital.}, 10:143--158, 2017.

\bibitem{bisi2021PhTB}
M.~Bisi.
\newblock Kinetic model for international trade allowing transfer of
  individuals.
\newblock {\em Phil. Trans. A}, 380:20210156. (pp. 1--14), 2022.

\bibitem{bisi2010PRE}
M.~Bisi, M.~Groppi, and G.~Spiga.
\newblock Kinetic {B}hatnagar-{G}ross-{K}rook model for fast reactive mixtures
  and its hydrodynamic limit.
\newblock {\em Phys. Rev. E}, 81:036327 (pp.~1--9), 2010.

\bibitem{bobylev2000JSP}
A.V. Bobylev, J.A. Carrillo, and I.~Gamba.
\newblock On some properties of kinetic and hydrodynamics equations for
  inelastic interactions.
\newblock {\em J. Stat. Phys.}, 98:743--773, 2000.

\bibitem{boffi1990PhysA}
V.C. Boffi, V.~Protopopescu, and G.~Spiga.
\newblock On the equivalence between the probabilistic, kinetic, and scattering
  kernel formulations of the {B}oltzmann equation.
\newblock {\em Physica A}, 164:400--410, 1990.

\bibitem{borsche2022PhysA}
R.~Borsche, A.~Klar, and M.~Zanella.
\newblock Kinetic-controlled hydrodynamics for multilane traffic models.
\newblock {\em Physica A}, 587:126486 (pp. 1--17), 2022.

\bibitem{borsoni2022CIMP}
T.~Borsoni, M.~Bisi, and M.~Groppi.
\newblock A general framework for the kinetic modelling of polyatomic gases.
\newblock {\em Comm. Math. Phys.}, 393:215–266, 2022.

\bibitem{burgerrossi}
M.~Burger and A.~Rossi.
\newblock Analysis of kinetic models for label switching and stochastic
  gradient descent.
\newblock Preprint: arXiv.2207.00389, 2022.

\bibitem{cercignani1988BOOK}
C.~Cercignani.
\newblock {\em The Boltzmann Equation and its Applications}.
\newblock Number~67 in Applied Mathematical Sciences. Springer, New York, 1988.

\bibitem{chapman1970BOOK}
S.~Chapman and T.G. Cowling.
\newblock {\em The Mathematical Theory of Non-Uniform Gases}.
\newblock Cambridge University Press, Cambridge, 1970.

\bibitem{cordier2005JSP}
S.~Cordier, L.~Pareschi, and G.~Toscani.
\newblock On a kinetic model for a simple market economy.
\newblock {\em J. Stat. Phys.}, 120(1):253--277, 2005.

\bibitem{delitala2014KRM}
M.~Delitala and T.~Lorenzi.
\newblock A mathematical model for value estimation with public information and
  herding.
\newblock {\em Kinet. Relat. Models}, 7(1):29--44, 2014.

\bibitem{dellamarca2022NHM}
R.~Della~Marca, N.~Loy, and A.~Tosin.
\newblock An {SIR}-like kinetic model tracking individuals' viral load.
\newblock {\em Networks and Heterogeneous Media}, 17:467--494, 2022.

\bibitem{dellamarca2022preprint}
R.~Della~Marca, N.~Loy, and A.~Tosin.
\newblock An {SIR} model with viral load-dependent transmission.
\newblock Preprint: arXiv:2208.12004, 2022.

\bibitem{dellamarca2022MMAS}
R.~Della~Marca, M.D.P. Machado~Ramos, C.~Ribeiro, and A.J. Soares.
\newblock Mathematical modelling of oscillating patterns for chronic autoimmune
  diseases.
\newblock {\em Math. Meth. Appl. Sci.}, 45:7144--7161, 2022.

\bibitem{dimarco2020PRE}
G.~Dimarco, L.~Pareschi, G.~Toscani, and M.~Zanella.
\newblock Wealth distribution under the spread of infectious diseases.
\newblock {\em Phys. Rev. E}, 102(2):022303, 2020.

\bibitem{During2010Econ}
B.~During.
\newblock Multi-species models in econo- and sociophysics.
\newblock {\em Econophysics \& economics of games, social choices and
  quantitative techniques (eds. B. Basu, B.K. Chakrabarti, S.R. Chakravarty, K.
  Gangopadhyay)}, pages 83--89, 2010.
\newblock Dordrecht: Springer.

\bibitem{during2009PRSA}
B.~During, P.~Markowich, J.F. Pietschmann, and M.T. Wolfram.
\newblock Boltzmann and {F}okker-{P}lanck equations modelling opinion formation
  in the presence of strong leaders.
\newblock {\em Proc. R. Soc. A}, 465:3687--3708, 2009.

\bibitem{DB2008CMS}
B.~During and G.~Toscani.
\newblock International and domestic trading and wealth distribution.
\newblock {\em Commun. Math. Sci.}, 6:1043–1058, 2008.

\bibitem{festa2018KRM}
A.~Festa, A.~Tosin, and M.T. Wolfram.
\newblock Kinetic description of collision avoidance in pedestrian crowds by
  sidestepping.
\newblock {\em Kinet. Relat. Models}, 11:491--520, 2018.

\bibitem{freguglia2017CMS}
P.~Freguglia and A.~Tosin.
\newblock Proposal of a risk model for vehicular traffic: {A} {B}oltzmann-type
  kinetic approach.
\newblock {\em Commun. Math. Sci.}, 15(1):213--236, 2017.

\bibitem{Gio1999}
V.~Giovangigli.
\newblock {\em Multicomponent Flow Modeling}.
\newblock Birkhäuser, Boston, 1999.

\bibitem{greenman2016PRE}
C.~D. Greenman and T.~Chou.
\newblock Kinetic theory of age-structured stochastic birth-death processes.
\newblock {\em Phys. Rev. E}, 93(1):012112, 2016.

\bibitem{herzberg1950BOOK}
G.~Herzberg.
\newblock {\em Molecular Spectra and Molecular Structure}.
\newblock Van Nostrand Reinold, New York, 1950.

\bibitem{kogan1969BOOK}
M.N. Kogan.
\newblock {\em Rarefied Gas Dynamics}.
\newblock Plenum Press, New York, 1969.

\bibitem{loy2021EJAM}
N.~Loy, T.~Hillen, and K.~Painter.
\newblock Direction dependent turning leads to anisotropic diffusion and
  persistence.
\newblock {\em European Journal of Applied Mathematics}, 33(4):729--765, 2022.

\bibitem{loy2020CMS}
N.~Loy and A.~Tosin.
\newblock Markov jump processes and collision-like models in the kinetic
  description of multi-agent systems.
\newblock {\em Commun. Math. Sci.}, 18(6):1539--1568, 2020.

\bibitem{loy2021KRM}
N.~Loy and A.~Tosin.
\newblock {B}oltzmann-type equations for multi-agent systems with label
  switching.
\newblock {\em Kinet. Relat. Models}, 14(5):867--894, 2021.

\bibitem{machadoramos2019JMB}
M.D.P. Machado~Ramos, C.~Ribeiro, and A.J. Soares.
\newblock A kinetic model of {T}-cell autoreactivity in autoimmune diseases.
\newblock {\em J. Math. Biol.}, 79:2005--2031, 2019.

\bibitem{pareschi2013BOOK}
L.~Pareschi and G.~Toscani.
\newblock {\em Interacting {M}ultiagent {S}ystems: {K}inetic equations and
  {M}onte {C}arlo methods}.
\newblock Oxford University Press, 2013.

\bibitem{pareschi2014PhTB}
L.~Pareschi and G.~Toscani.
\newblock Wealth distribution and collective knowledge: a {B}oltzmann approach.
\newblock {\em Phil. Trans. A}, 372:20130396, 2014.

\bibitem{pirner2021Fluids}
M.~Pirner.
\newblock A review on {BGK} models for gas mixtures of mono and polyatomic
  molecules.
\newblock {\em Fluids}, 6:393, 2021.

\bibitem{puppo2017KRM}
G.~Puppo, M.~Semplice, A.~Tosin, and G.~Visconti.
\newblock Kinetic models for traffic flow resulting in a reduced space of
  microscopic velocities.
\newblock {\em Kinet. Relat. Models}, 10(3):823--854, 2017.

\bibitem{rossani1999PhysA}
A.~Rossani and G.~Spiga.
\newblock A note on the kinetic theory of chemically reacting gases.
\newblock {\em Physica A}, 272:563--573, 1999.

\bibitem{spiga1985PhysA}
G.~Spiga, T.~Nonnenmacher, and V.C. Boffi.
\newblock Moment equations for the diffusion of the particles of a mixture via
  the scattering kernel formulation of the nonlinear {B}oltzmann equation.
\newblock {\em Physica A}, 131:431--448, 1985.

\bibitem{toscani2006CMS}
G.~Toscani.
\newblock Kinetic models of opinion formation.
\newblock {\em Commun. Math. Sci.}, 4(3):481--496, 2006.

\bibitem{toscani2018PRE}
G.~Toscani, A.~Tosin, and M.~Zanella.
\newblock Opinion modeling on social media and marketing aspects.
\newblock {\em Phys. Rev. E}, 98:022315, 2018.

\bibitem{toscani2019PRE}
G.~Toscani, A.~Tosin, and M.~Zanella.
\newblock Multiple-interaction kinetic modeling of a virtual-item gambling
  economy.
\newblock {\em Phys. Rev. E}, 100:012308, 2019.

\bibitem{villani1998ARMA}
C.~Villani.
\newblock On a new class of weak solutions to the spatially homogeneous
  {B}oltzmann and {L}andau equations.
\newblock {\em Arch. Ration. Mech. Anal.}, 143(3):273--307, 1998.

\bibitem{waldamnn1958HdP}
L.~Waldmann.
\newblock Transporterscheinungen in gasen von mittlerem druck.
\newblock {\em Handbuch der Physik}, 12:295--514, 1958.
\newblock S. Fl\"ugge ed., Springer Verlag, Berlin.

\end{thebibliography}

\appendix

\end{document}